%% file: main.tex
\documentclass[11pt]{article}

\usepackage[utf8]{inputenc}

\usepackage{macro_yihan}

\title{Tight Bounds on List-Decodable and List-Recoverable\\Zero-Rate Codes}

\author{
    Nicolas Resch\thanks{University of Amsterdam. Email: \href{mailto:n.a.resch@uva.nl}{\texttt{n.a.resch@uva.nl}}.}
    \and
    Chen Yuan\thanks{Shanghai Jiao Tong University. Email: \href{mailto:chen_yuan@sjtu.cn.edu}{\texttt{chen\_yuan@sjtu.cn.edu}}.}
    \and
    Yihan Zhang\thanks{Institute of Science and Technology Austria. Email: \href{mailto:zephyr.z798@gmail.com}{\texttt{zephyr.z798@gmail.com}}.}
}

\begin{document}

\maketitle

\begin{abstract}
    In this work, we consider the list-decodability and list-recoverability of codes in the \emph{zero-rate regime}. Briefly, a code $\cC \subseteq [q]^n$ is $(p,\ell,L)$-list-recoverable if for all tuples of input lists $(Y_1,\dots,Y_n)$ with each $Y_i \subseteq [q]$ and $|Y_i|=\ell$ the number of codewords $c \in \cC$ such that $c_i \notin Y_i$ for at most $pn$ choices of $i \in [n]$ is less than $L$; list-decoding is the special case of $\ell=1$. In recent work by Resch, Yuan and Zhang~(ICALP~2023) the zero-rate threshold for list-recovery was determined for all parameters: that is, the work explicitly computes $p_*:=p_*(q,\ell,L)$ with the property that for all $\eps>0$ (a) there exist infinite families positive-rate $(p_*-\eps,\ell,L)$-list-recoverable codes, and (b) any $(p_*+\eps,\ell,L)$-list-recoverable code has rate $0$. In fact, in the latter case the code has \emph{constant} size, independent on $n$. However, the constant size in their work is quite large in $1/\epsilon$, at least $|\cC|\geq (\frac{1}{\epsilon})^{O(q^L)}$.

    Our contribution in this work is to show that for all choices of $q,\ell$ and $L$ with $q \geq 3$, any $(p_*+\eps,\ell,L)$-list-recoverable code must have size $O_{q,\ell,L}(1/\eps)$, and furthermore this upper bound is complemented by a matching lower bound $\Omega_{q,\ell,L}(1/\eps)$. This greatly generalizes work by Alon, Bukh and Polyanskiy~(IEEE Trans.\ Inf.\ Theory~2018) which focused only on the case of binary alphabet (and thus necessarily only list-decoding). We remark that we can in fact recover the same result for $q=2$ and even $L$, as obtained by Alon, Bukh and Polyanskiy: we thus strictly generalize their work. 

    Our main technical contribution is to (a) properly define a linear programming relaxation of the list-recovery condition over large alphabets; and (b) to demonstrate that a certain function defined on a $q$-ary probability simplex is maximized by the uniform distribution. This represents the core challenge in generalizing to larger $q$ (as a $2$-ary simplex can be naturally identified with a one-dimensional interval). We can subsequently re-utilize certain Schur convexity and convexity properties established for a related function by Resch, Yuan and Zhang along with ideas of Alon, Bukh and Polyanskiy. 

\end{abstract}

\tableofcontents

\section{Introduction} \label{sec:intro}
\input{intro}

\section{Preliminaries} \label{sec:prelims}
\input{prelims}

\section{Zero-Rate List-Decoding}
\subsection{Linear Programming Relaxation} \label{sec:ham-to-euc}
\input{list-dec-radius}
\subsection{Properties of $f(P,\omega)$} \label{sec:fun-props}
\input{list-dec-function-props}
\subsection{Abundance of Random-Like $L$-tuples} \label{sec:abundance}
\input{list-dec-abundance}
\subsection{Putting Everything Together} \label{sec:end}
\input{list-dec-end}

\section{Zero-Rate List-Recovery} \label{sec:list-rec}
In this section, we show how our results on list-decoding can naturally be extended to list-recovery. As many of the ideas are the same, we mostly focus upon indicating the changes that need to be made for this more general setting. 
\subsection{Linear Programming Relaxation} \label{sec:radius}
\input{list-rec-radius}
\subsection{Properties of $f_\ell(P,\omega)$} \label{sec:property}
\input{list-rec-function-property}

\subsection{Abundance of Random-Like $L$-tuples} \label{sec:lrabundance}
\input{list-rec}
\subsection{Putting Everything Together} \label{sec:lrend}
\input{list-rec-end}

\section{Code Construction} \label{sec:code-construct}
\input{list-dec-constr}

\section{Conclusion} \label{sec:conclusion}
We end the paper with a few concluding remarks and open questions. 
\begin{itemize}
    \item Due to the use of hypergraph Ramsey's theorem on page \pageref{ramsey-listrec}, our upper bound in \Cref{thm:list-rec-main} is valid only for very small $\eps$. 
    The question of determining the optimal code size for any $0<\eps\le 1 - p_*(q, \ell, L)$ remains open. 

    \item Though our upper and lower bounds match in terms of the order of $ 1/\eps $, the hidden constants (in particular their dependence on $ q, \ell, L $) are rather different. 
    Our construction gives an explicit constant $ c_{q, \ell, L} $ (see \Cref{eqn:const_constr}). 
    It is possible that such codes have optimal size as a function of the gap-to-zero-rate-threshold even in terms of the pre-factor. 
    However, we are not sufficiently confident to make this a conjecture. 
    On the other hand, the constant in our upper bound is implicit and is not expected to be close to the lower bound even if made explicit. 
    Studying the leading coefficient of the maximal size of zero-rate codes requires additional ideas. 

     \item It would be interesting to see how techniques developed in this work can be used to study zero-rate codes under other metrics such as the Lee metric \cite{TB_lee}, $ \ell_1 $ metric \cite{TB11,TB12} and others for which the zero-rate threshold can be determined by the double counting argument \cite{ahlswede-blinovsky-metric}. 
\end{itemize}

\bibliographystyle{alpha}
\bibliography{ref} 

\appendix
\section{Auxiliary lemmas} \label{app:aux}
\begin{proposition}[Basic feasible solution, {\cite[Section 4.2]{gartner-matousek-LP-book}}]
\label{prop:basic_feasible_sol}
For given $ \vc\in\bbR^n, \vb\in\bbR^m, A\in\bbR^{m\times n} $, consider a linear program $\mathsf{LP}$ in equational form (without loss of generality): 
\begin{align}
    \textnormal{maximize} \quad \inprod{\vc}{\vx} , \qquad 
    \textnormal{subject to} \quad A \vx = \vb , \vx \ge 0 , \notag 
\end{align}
where the last constraint means that $\vx$ is element-wise non-negative. 
Suppose without loss of generality that $ m \le n $ and $ \rk(A) = m $. 
Then there exists at least one solution $ \vx^*\in\bbR^n $, known as a \emph{basic feasible solution}, with at most $m$ nonzero elements (i.e., at least $n-m$ zeros). 
Furthermore $ \vx^* $ is determined only by $ (A, \vb) $, independent of $ \vc $. 
If $\mathsf{LP}$ has an optimal solution then it has an optimal basic feasible solution. 
\end{proposition}

\end{document}

%% file: intro.tex
Given an error-correcting code $\mathcal C \subseteq [q]^n$, a fundamental requirement is that the codewords are sufficiently well-spread in order to guarantee some non-trivial correctability properties. This is typically enforced by requiring that the minimum distance of the code $d = \min\{\disth(\vec c, \vec c') : \vec c \neq \vec c' \in \mathcal C\}$, where $\disth(\cdot,\cdot)$ denotes the Hamming distance (i.e. the number of coordinates on which two strings differ). Note that minimum distance $d$ is equivalent to the following ``packing'' property: if we put a ball of radius $r:=\lfloor d/2\rfloor$ around any point $\vec z \in [q]^n$ -- i.e. we consider the Hamming ball $\ballh(\vec y,r):=\{\vec x \in [q]^n:\disth(\vec x,\vec y)\leq r\}$ -- then all these balls contain at most 1 codeword from $\mathcal C$.

This latter viewpoint can easily be generalized to obtain \emph{list-decodability}, where we now require that such Hamming balls do not capture ``too many'' codewords. That is, for $p \in [0,1]$ and $L \in \bbN$ a code is called $(p,L)$\emph{-list-decodable} if every Hamming ball of radius $pn$ contains less than $L$ codewords from $\mathcal C$. In notation: for all $\vy \in [q]^n$, $|\ballh(\vy,pn)|\leq L-1$.\footnote{Typically the upper bound is $L$, rather than $L-1$. However, for ``impossibility'' arguments this parametrization is more common, as it leads to less cumbersome computations.} This notion was already introduced in the 50's by Elias and Wozencraft~\cite{Elias57,Wozencraft58,elias1991error} but has in the past 20 years seen quite a bit of attention due to its connections to other parts of theoretical computer science~\cite{goldreich1989hard,babai1990bpp,lipton1990efficient,kushilevitz1993learning,jackson1997efficient,sudan2001pseudorandom}.

One can push this generalization further to obtain \emph{list-recoverability}. Here, we consider a tuple of input lists $\bY = (Y_1,\dots,Y_n)$, where each $Y_i \subseteq [q]$ has size at most $\ell$ (for some $\ell \in \bbN$). The requirement is that the number of codewords that ``disagree'' with $\bY$ in at most $pn$ coordinates is at most $L-1$. More formally, if for all $\bY = (Y_1,\dots,Y_n)$ the number of codewords $\vc \in \cC$ such that $|\{i \in [n]:c_i \notin Y_i\}| \leq pn$ is at most $L-1$, the code is called $(p,\ell,L)$\emph{-list-recoverable}. Note that $(p,L)$-list-decodability is nothing other than $(p,1,L)$-list-recoverability. Initially, list-recoverability was abstracted as a useful stepping stone towards list-decoding concatenated codes. However, in recent years this notion has found many connections to other parts of computer, e.g. in cryptography~\cite{HIOS15,holmgren2021fiat}, randomness extraction~\cite{guruswami2009unbalanced}, hardness amplification~\cite{doron2020nearly}, group testing~\cite{INR10,NPR12}, streaming algorithms~\cite{doron2022high}, and beyond.

\paragraph{Rate versus noise-resilience.} Having fixed a desired ``error-tolerance'' as determined by the parameters $p,\ell$ and $L$ we would also like the code $\mathcal C$ to be as large as possible: intuitively, this implies that the code contains the minimal amount of redundancy possible. A fundamental question in coding theory is to understand the achievable tradeoffs between the rate $R := \frac{\log_q |\mathcal{C}|}{n}$ and some ``error-resilience'' property of the code, e.g., minimum distance, list-decodability, or list-recoverability. 

This question in full generality is wide open. Even for the special case of $q=2$ and $L=2$ (i.e. determining the optimal tradeoff between rate and distance for binary codes) is unclear: on the possibility side we have the Gilbert-Varshamov bound~\cite{gilbert1952comparison,varshamov1957estimate} showing $R \geq 1-H_2(p/2)$ is achievable (here, $H_2(x)=-x\log_2x-(1-x)\log_2(1-x)$ is the binary entropy function), while bounds of Elias and Bassalygo~\cite{bassalygo-pit1965} and the linear programming bound~\cite{mrrw1,mrrw2,delsarte-1973} give incomparable and non-tight upper bounds. None of these bounds have been substantially improved in at least 40 years. The situation is even more complicated for larger $q$: for $q=49$ (and larger prime powers) the celebrated algebraic geometry codes of Tsafsman, Vladut and Zink~\cite{tsfasman1982modular} provide explicit codes of higher rate in certain regimes than those promised by the Gilbert-Varshamov bound. 

When one relaxes the question to allow an asymptotically growing list size $L$ then we do have a satisfactory answer: the answer is provided by the list-decoding/-recovery theorem, which states that for all $\eps>0$ there exist $(p,\ell,O(1/\eps))$-list-recoverable codes of rate $1-H_{q,\ell}(p)$ where
\[
    H_{q,\ell}(x):= p\log_q\left(\frac{q-\ell}{p}\right) + (1-p)\log_q\left(\frac{\ell}{1-p}\right)
\]
is $(q,\ell)$-ary entropy function~\cite{resch2020thesis}.\footnote{Note that setting $\ell=1$ recovers the standard $q$-ary entropy function, which itself reduces to the binary entropy function upon setting $q=2$.} On the other hand, any code of rate $R \geq 1-h_{q,\ell}(p)$ fails to be $(p,\ell,L)$-list-recoverable unless $L \geq q^{\Omega(\eps n)}$. However, this does not provide very meaningful bounds if one is interested in, say, $(p,2,5)$-list-recoverable codes.


\paragraph{Positive versus zero-rate regimes.} Thus far, we have implicitly been discussing the \emph{positive-rate regime}. However, one can also ask questions about the behaviour of codes in the \emph{zero-rate regime}. For context, recent work by Resch, Yuan and Zhang~\cite{RYZ22} computed the zero-rate threshold for list-recovery: that is, for all alphabet sizes $q \geq 2$, input list sizes $\ell$ and output list size $L$, they determine the value $p_*(q,\ell,L)$ such that (a) for all $p<p_*(q,\ell,L)$ there exist infinite families of positive rate $(p,\ell,L)$-list-recoverable codes over the alphabet $[q]$, and (b) for all $p>p_*(q,\ell,L)$ there does not exist such an infinite family. 

Having now delineated the ``positive rate'' and the ``zero-rate'' regimes depending on how $p$ compares to $p_*(q,\ell,L)$, in this work we study the zero-rate regime for list-recoverable codes for all alphabet sizes $q$. In \cite{RYZ22}, it is shown that $(p,\ell,L)$-list-recoverable codes $\cC \subseteq [q]^n$ with $p=p_*(q,\ell,L)+\eps$ have constant size (that is, independent of the block length $n$); however, this constant is massive in the parameters due to the use of a Ramsey-theoretic bound. In particular, the dependence on $\eps$ is at least $(1/\eps)^{2q^L}$, and this is additionally multiplied by a tower of 2's of height roughly $L$. 

To the best of our knowledge, prior work on this question focuses exclusively on the $q=2$ case. For example, in the case of $L=2$ (i.e., unique-decoding) we have $p_*(2,1,2)=1/4$, and work by Levenshtein shows A particularly relevant prior work is due to Alon, Bukh and Polyanskiy~\cite{abp-2018}. Herein the authors consider this question for the special case of $q=2$ (and thus, necessarily, only for list-decoding). In particular, they show that when $L$ is even if $p = p_*(2,1,L)+\eps$ then such a $(p,L)$-list-decodable code $\cC \subseteq [2]^n$ has size at most $O_L(1/\eps)$, and moreover provide a construction of such a code with size $\Omega_L(1/\eps)$.\footnote{Note that for the special case of $q=2$, the zero-rate threshold for list-decoding had already been established by Blinovsky~\cite{blinovsky-1986-ls-lb-binary}.} They observe some interesting behaviour in the case of odd $L$; in particular, the maximum size of a $(p_*(2,1,3)+\eps,3)$-list-decodable code is $\Theta(1/\eps^{3/2})$.\footnote{This argument in fact shows a flaw in an earlier claimed proof of Blinovsky that claimed such codes have size $O_L(1/\eps)$ for all $L \in \bbN$.}

Our motivations for this investigation are three-fold. Firstly, the zero-rate regime offers combinatorial challenges and interesting behaviours that we uncover in this work. Secondly, many codes that find applications in other areas of theoretical computer in fact have subconstant rate. Lastly, the zero-rate regime appears much more tractable than the positive rate regime -- indeed, we can obtain tight upper and lower bounds on the size of a code, as we will soon see. It would be interesting to determine to what extent such techniques could be useful for understanding the positive rate regime as well. 

\subsection{Our results.}

Our main result in this work is a tight bound on the size of a $(p,\ell,L)$-list-recoverable code over an alphabet of size $q \geq 3$ when $p > p_*(q,\ell,L)$. The main technical challenge is to compute the following upper bound on the size of such a code. 


\begin{theorem} [Informal Version of \Cref{thm:list-rec-main}] \label{thm:informal-ub}
Let $q,\ell,L \in \bbN$ with $q \geq 3$. $\ell <q$ and $L>\ell$ be fixed constants. Let $\eps>0$ and put $p = p_*(q,\ell,L)+\eps$. Suppose $\cC \subseteq [q]^n$ is $(p,\ell,L)$-list-recoverable. Then $|\cC| \leq O_{q,\ell,L}(1/\eps)$.
\end{theorem}

We complement the above negative result with the following code construction, showing the upper bound is tight. 

\begin{theorem} [Informal Version of \Cref{thm:list-dec-qary-postplotkin-constr}] \label{thm:informal-lb}
Let $q,\ell,L \in \bbN$ with $q \geq 3$ and $\ell < q$ be fixed constants. Let $\eps>0$ and put $p = p_*(q,\ell,L)+\eps$. There exists a $(p,\ell,L)$-list-recoverable code $\cC \subseteq [q]^n$ such that $|\cC| \geq \Omega_{q,\ell,L}(1/\eps)$.
\end{theorem}

We emphasize that in the above theorems the implied constants may depend on $q,\ell$ and $L$. 

Note that our results explicitly \emph{exclude} the case of $q=2$. As \cite{abp-2018} prove, the binary alphabet behaves in subtle ways: the bound on the code size depends on the parity of $L$. Intriguingly, our work demonstrates that such behaviour does not arise over larger alphabets. 

\subsection{Technical Overview} \label{subsec:technical-overview}

\paragraph{The double-counting argument.}
Since our focus is on zero-rate list-decodable/-recoverable codes, it helps to first review the proof of the zero-rate threshold $ p_*(q, \ell, L) $. 
A lower bound can be easily obtained by a random construction that attains a positive rate for any $ p \le p_*(q, \ell, L) - \eps $. 
For the upper bound, let us first consider the list-decoding case, i.e., $\ell=1$. 
The proof in \cite{blinovsky-2005-ls-lb-qary,blinovsky-2008-ls-lb-qary-supplementary,RYZ22}, at a high-level, proceeds via a double-counting argument.\footnote{A characterization of $ p_*(q, 1, L) $ was announced in \cite{blinovsky-2005-ls-lb-qary,blinovsky-2008-ls-lb-qary-supplementary} whose proof was flawed. The work \cite{RYZ22} filled in the gaps therein and characterized $ p_*(q, \ell, L) $ for general $\ell$.} 
For any $(p, \ell, L)$-list-decodable code $ \cC\subset[q]^n $, the proof aims to upper and lower bound the \emph{radius} of a list averaged over the choice of the list from $\cC$:
\begin{align}
    \frac{1}{M^L} \sum_{(\vc_1, \cdots, \vc_L)\in\cC^L} \radh(\vc_1, \cdots, \vc_L) . \label{eqn:double-count} 
\end{align}
Comparing the bounds produces an upper bound on $|\cC|$. 
Here $ \radh(\cdot) $, known as the Chebyshev radius of a list, is the relative radius of the smallest Hamming ball containing all codewords in the list. 
A lower bound on \Cref{eqn:double-count} essentially follows from list-decodability of $\cC$. Indeed, each term (corresponding to lists consisting of distinct codewords) is lower bounded by $p$, otherwise a list that fits into a ball of radius at most $np$ is found, violating list-decodability of $\cC$. 
Therefore \Cref{eqn:double-count} is at least $p-o(1)$, where $o(1)$ is to account for lists with not-all-distinct codewords. 

On the other hand, it is much more tricky to upper bound \Cref{eqn:double-count} as, in general, $\radh$ admits no analytically closed form and can only be computed by solving a min-max problem. 
Previous proofs \cite{RYZ22} then first extracts a subcode $\cC'$ with highly-regular list structures via the hypergraph Ramsey's theorem. 
This allows one to assert that all lists have essentially the same radius and all codewords in each list have essentially the same distance to the center of the list. 
As a result, the min-max expression is ``linearized'' and \Cref{eqn:double-count} can be upper bounded when restricted to $\cC'$. 
The downside is that the Ramsey reduction step is rather lossy for code size. 

\paragraph{Weighted average radius.}
The effect of the Ramsey reduction, put formally, is to enforce the \emph{average radius}:
\begin{align}
    \ol{\rad}_\ham(\vc_1, \cdots, \vc_L) &\coloneqq \frac{1}{n} \min_{\vr\in\{0,1\}^n} \frac{1}{L} \sum_{i = 1}^L \disth(\vc_i, \vr) \label{eqn:avg_rad_intro} 
\end{align}
of every list in the subcode to be approximately equal. 
To extract the regularity structures in lists without resorting to extremal bounds from Ramsey theory, \cite{abp-2018} introduced the notion of \emph{weighted average radius} which ``linearizes'' the Chebyshev radius in a \emph{weighted} manner:
\begin{align}
    \ol{\rad}_{\omega}(\vc_1, \cdots, \vc_L) &\coloneqq \frac{1}{n} \min_{\vr\in\{0,1\}^n} \sum_{i = 1}^L \omega(i) \disth(\vc_i, \vr) \notag 
\end{align}
where $ \omega $ is a distribution on $L$ elements. 
For any weighting $\omega$, $ \ol{\rad}_\omega $ of lists from the code forms a suite of succinct statistics of the list distribution. 
It turns out $ \ol{\rad}_{U_L} = \ol{\rad} $ (where $ U_L $ denotes the uniform distribution on $ [L] $) is maximal under all $\omega$. 
Recall that the double-counting argument suggests that in an optimal zero-rate code, the behaviour of the ensemble average of $\rad$ is essentially captured by that of $ \ol{\rad} $. 
In particular, list-decodability ensures that $ \ol{\rad} $ of \emph{most} lists should be large. 
However, not too many lists in an optimal code are expected to have large $ \ol{\rad}_\omega $ for any $ \omega\ne U_L $. 
\cite{abp-2018} then managed to quantify the gap between $ \ol{\rad} = \ol{\rad}_{U_L} $ and $ \ol{\rad}_\omega $ (with $ \omega\ne U_L $), which yields improved (and sometimes optimal) size-radius trade-off of zero-rate codes. 

\paragraph{Generalization to $q$-ary list-decoding.}
Our major technical contribution is in extrapolating the above ideas to list-recovery. 
The challenge lies particularly in defining a proper notion of weighted average radius and proving its properties. 
Our definition relies crucially on an embedding $\phi$ from $[q]$ to the simplex in $ \bbR^q $ and relaxes the center $\vr$ of the list to be a fractional vector. 
Specifically, denoting by $ \Delta $ the simplex in $ 
\bbR^q $ and $ \partial\Delta = \{\ve_1, \cdots, \ve_q\} $ its vertices (i.e., the standard basis of $ \bbR^q $), we let the embedding $\phi$ map each symbol $x\in[q]$ to the one-hot vector $\ve_x\in\partial\Delta$. 
Denoting by $ \vx_1, \cdots, \vx_L\in(\partial\Delta)^n $ the (element-wise) images of a list $\vc_1, \cdots, \vc_L\in[q]^n$, we define the \emph{weighted average radius} of $ \vx_1, \cdots, \vx_L $ as:
\begin{align}
    \ol{\rad}_\omega(\vx_1, \cdots, \vx_L) &= \frac{1}{n} \min_{\vy\in\Delta^n} \frac{1}{2} \exptover{i\sim\omega}{\normone{ \vx_i - \vy }} , \label{eqn:weighted-avg-rad-intro} 
\end{align}
where $ \omega $ is any distribution on $ [L] $. 

The notriviality and significance of the above notion, especially the embedding used therein, is three-fold. 
\begin{itemize}
    \item First, as the weighting $\omega$ varies, $ \ol{\rad}_\omega $ serves as a bridge between the standard average radius in \Cref{eqn:avg_rad_intro} and the Chebyshev radius. 
    Indeed, $\omega = U_L$ recovers the former, and the maximum $ \ol{\rad}_\omega $ over $ \omega $ recovers the latter. 
    However, we caution that the second statement does not hold without the embedding since the Hamming distance between $q$-ary symbols per se is not and cannot be interpolated by a convex function, which makes the minimax theorem inapplicable. 
    Fortunately, our embedding affinely extend the $q$-ary Hamming distance to the simplex therefore brings back the applicability of the minimax theorem and connects $ \max_\omega \ol{\rad}_\omega $ to $ \rad $. 

    \item Second, our definition in \Cref{eqn:weighted-avg-rad-intro} allows $\vy$ to take any value on the simplex, instead of only its vertices, i.e., the image of $[q]$ under $\phi$. 
    Though embedding naively to the hypercube $ [0,1]^q $ seems convenient, upon solving the expression with fractional $\vy$ one does not necessarily obtain a notion that is guaranteed to closely approximate the original version with integral $\vy$. 
    In contrast, using linear programming duality, we show that our embedding yields relaxed notion of radius which closely approximates the actual Chebyshev radius. 
    Indeed, upon rounding the fractional center $\vy$ and taking its pre-image under $\phi$, our results guarantee that the resulting radius must have negligible difference from the Chebyshev radius. Precisely speaking, we want to find a vector $\vy=(y(i,j))_{[n]\times [q]}\in \Delta$ close to the $L$ images of the codewords $\vx_1,\ldots,\vx_L$ by linear programming. Meanwhile, we want $\vy(i):=(y(i,1),\ldots,y(i,q))$ to belong to $\partial \Delta$ so that we can find a preimage of $\vy(i)$ in $[q]$.  
    Since $\vy(i)\in \Delta$, the components in $\vy(i)$ are subject to $\sum_{j=1}^q y(i,j)=1$. This implies that at least one component of $\vy(i)$ is nonzero.
    The basic feasible solution in \Cref{prop:basic_feasible_sol} guarantees that there exists a feasible solution such that most of $y(i,j)$ are $0$. Combining with the fact $\sum_{j=1}^q y(i,j)=1$ forces $(y(i,1),\ldots,y(i,q))\in \partial \Delta$ for almost all $i\in [n]$. Thus, we obtain a negligible loss in the conversion between Hamming distance and Euclidean distance. 
    \item Finally, under the embedding $\phi$, the weighted average radius $ \ol{\rad}_\omega $ still retains the appealing feature that the minimization can be analytically solved, therefore giving rise to an explicit expression (see \Cref{eqn:wt_avg_rad_2}) which greatly facilitates our analysis. 
    
\end{itemize}

We then show, via techniques deviating from those in \cite{abp-2018}, three key properties that are required by the subsequent arguments.
\begin{enumerate}
    \item \label{itm:property1} For any fixed distribution $P$, if entries of codewords in the list are generated i.i.d.\ using $P$, then 
    $$f(P, \omega) \coloneqq \exptover{(X_1, \cdots, X_L)\sim P^{\ot L}}{1 - \max_{x\in[q]} \sum_{\substack{i\in[L] \\ X_i = x}} \omega(i)}$$ 
    is maximized when $ \omega = U_L $. Moreover, the equality holds if and only if $\omega=U_L$ for $q\geq 3$ and any $L$. Our approach is different from \cite{abp-2018} as we can not explicitly represent function $f(P,\omega)$. 

    \item \label{itm:property2} Furthermore, if entries of codewords in the list are generated i.i.d.\ using a certain $P$, then $ f(P, U_L)$ is upper bounded by $f(P_{q,p}, U_L)$  with $P_{q,p}=(\frac{1-p}{q},\ldots,\frac{1-p}{q},p)$ and $p=\max_{i\in [q]} P(i)$. This follows from the Schur convexity property proved in \cite{RYZ22}.
    \item \label{itm:property3} Finally, denoting by $P_i$ the distribution of the $i$-th components of codewords in code $\cC$, Schur convexity promises $f(P_i,U_L)\leq f(P_{q,p_i}, U_L)$. In \cite{RYZ22}, it is proved that $f(P_{q,p}, U_L)$ is convex for $p\in [1/q, 1]$. Thus, we can conclude that 
    $$
    \frac{1}{n}\sum_{i\in [n]}f(P_i,U_L)\leq f(P_{q,p},U_L)
    $$
    with $p=\frac{1}{n}\sum_{i\in [n]}p_i$. 
\end{enumerate}

The remaining part of our proof is similar to \cite{abp-2018}. We show that a code $\cC$ either has radius 
$$\rad(\cC)=\frac{1}{n}\min_{\vx\in [q]^n}\max_{\vc\in \cC} \disth(\vc,\vx)\leq 1-\frac{1}{q}-\delta$$
or most of $L$-tuples with distinct codewords in $\cC$ are distributed close to uniform,. For the former case, we use the convexity property to show that the list-decodability of $\cC$ can not exceed $f(U_q, U_L)=p_*(q,L)$ by much. For the latter case, since most of $L$-tuples of distinct codewords in $\cC^L$ looks uniformly at random, we can show that the list-decodablilty of $\cC$ is very close to that of random codes which is $f(U_q,U_L)$.
\paragraph{Generalization to list-recovery.} For list-recovery, i.e., $ \ell>1 $, we find an embedding $\phi_\ell$ that maps each element in $[q]$ to a superposition of $\ell$ vertices of the simplex in $ \bbR^q $, i.e., we map the element in $[q]$ to a vector space $[0,1]^{\cX}$ where $\cX=\binom{[q]}{\ell}$ is the collection of all $\ell$-subsets in $[q]$. Concretely, we define $\phi_\ell(i):=\sum_{A\in \cX, i\in A}\ve_A$ where $(\ve_A)_{A\in \cX}$ is a standard basis of $\bbR^{\cX}$. The intuition behind this map is that if $i\in X$, we have $\normone{ \phi_\ell(i) - \ve_X }=\binom{q}{\ell}-1$  and otherwise $\normone{ \phi_\ell(i) - \ve_X }=\binom{q}{\ell}+1$. 
Similar to the list decoding, given $L$ codewords in $[q]^n$, we obtain $L$ vectors $\vx_1,\ldots,\vx_L$ under the map $\phi_\ell$. Our goal is to find a vector $\vy=(y(i,A))_{[n]\times \cX}$ close to these $L$ vectors subject to the constraint that $\sum_{A\in \cX}y(i,A)=1$ for any $i\in [n]$. This constraint combined with the basic feasible solution argument in \Cref{prop:basic_feasible_sol} forces that for almost all $i\in [n]$, $(y(i,A))_{A\in \cX}$ is of the form $\ve_X$. For such $i$, we can find an $\ell$-subset $X\in \cX$ preserving the distance, i.e., 
$$\distlr(i,X)=\indicator{i\notin X}=\frac{1}{2}\left(\normone{ \phi_\ell(i) - \ve_X }-\binom{q}{\ell}+1\right).$$   
Besides the linear programming relaxation, further adjustments for the proof of properties analogous to \Cref{itm:property1,itm:property2,itm:property3} above are required.

\paragraph{Code construction.}  As alluded to before, a code that saturates the optimal size-radius trade-off should essentially saturate both the upper and lower bounds on the quantity
\begin{align}
    \frac{1}{M^L} \sum_{(\vc_1, \cdots, \vc_L)\in\cC^L} \ol{\rad}_\ham(\vc_1, \cdots, \vc_L) \notag 
\end{align}
considered in the double-counting argument. 
Indeed, our impossibility result implies that any optimal zero-rate code must contain a large fraction of random-like $L$-tuples $(\vc_1,\ldots,\vc_L)$, i.e., for every $\vu\in [q]^L$ \begin{equation}\label{eqn:type}
\sum_{i=1}^{n}\indicator{(\vc_1(i),\ldots,\vc_L(i))=\vu}\approx \frac{n}{q^L}
\end{equation}
where $\vc_j=(\vc_j(1),\ldots,\vc_j(n))\in [q]^n$.
To match such an impossibility result, an optimal construction should contain as many such $L$-tuples as possible. 
A simplex-like code then becomes a natural candidate. 
This is a natural extension of the construction in \cite{abp-2018} to larger alphabet. 
An $ M\times n $ codebook $\cC$ consisting of $M$ codewords each of length $n$ is constructed by putting as columns all possible distinct length-$M$ vectors that contains identical numbers of $ 1, 2, \cdots, q $. 
It is not hard to see by symmetry that \eqref{eqn:type} becomes equality for every $L$-tuple with distinct codewords in $\cC$. Thus, $\cC$ is the most regular code.

We also remark that, unlike for positive-rate codes, the prototypical random construction (with expurgation) does not lead to favorable size-radius trade-off since the deviation of random sampling is comparatively too large in the zero-rate regime. 
In contrast, the simplex code is deterministically regular and has no deviation. 

\subsection{Organization} \label{subsec:organization}
The remainder is organized as follows. First, \Cref{sec:prelims} provides the necessary notations and definitions, together with some preliminary results which will be useful in the subsequent arguments. \Cref{sec:ham-to-euc,sec:fun-props,sec:abundance,sec:end} contain our argument establishing \Cref{thm:informal-ub} for list-decoding (i.e. the case $\ell=1$); in \Cref{sec:list-rec} we elucidate the changes that need to be made to establish the theorem for general $\ell$. Next, \Cref{sec:code-construct} provides the code construction establishing \Cref{thm:informal-lb}. We lastly summarize our contribution in \Cref{sec:conclusion} and state open problems. 

%% file: prelims.tex
Firstly, for convenience of the reader we begin by summarizing the notation that we use. This is particularly relevent as we will often be in situations where we need multiple indexes for, e.g., lists of vectors where each coordinate lies in a probability simplex, so the reader is encouraged to refer to this table whenever it is unclear what is intended. 
\begin{table}[htbp]
    \centering
    \begin{tabular}{c|c}
        English letter in boldface & $[q]^n$-valued vector \\
        Greek letter in boldface & $\Delta([q])$-valued vector \\
        $\Delta:= \Delta([q])$ &Simplex in $[0,1]^q$, i,e., $\Delta=\{(x_1,\ldots,x_q)\in [0,1]^q: \sum_{i=1}^q x_i=1\}$ \\
        $\partial\Delta$ & Set of vertices of $\Delta$ \\
        $\ve_x\in\partial\Delta$ & The image of $x\in[q]$ under $\phi$, i.e., the $x$-th vertex of $\Delta$ \\
        $\vc_i\in[q]^n$ & The $i$-th codeword in a list \\
        $\vx_i\in\Delta^n$ & Image of $\vc_i$ under $\phi$ (applied component-wise) \\
        $\vy\in\Delta^n$ & Relaxed center of a list \\
        $\vx(j)\in\partial\Delta, \vy(j)\in\Delta$ & The $j$-th block (of length $q$) in $\vx\in(\partial\Delta)^n, \vy\in\Delta^n$, respectively \\
        $x(j,k)\in\{0,1\}, y(j,k)\in[0,1]$  & The $(j,k)$-th element of $\vx\in(\partial\Delta)^n, \vy\in\Delta^n$, respectively \\
        $\radh$ & (Standard) Chebyshev radius \\
        $\rad$ & Relaxed Chebyshev radius \\
        $\ol{\rad}$ & Average radius \\
        $\ol{\rad}_\omega$ & Average radius weighted by $\omega\in\Delta([L])$ \\
        $f(P, \omega)$ & Expected average radius (weighted by $\omega$) of $P$-distributed symbols \\
        $(X_1, \cdots, X_L)\sim P^{\ot L}$ & A list of i.i.d.\ $P$-distributed symbols \\
        $U_k$ & Uniform distribution on $[k]$
    \end{tabular}
    \caption{Notation for list-decoding. }
    \label{tab:listrec-notation}
\end{table}

For a finite set $S$ and an integer $ 0\le k\le|S| $, we denote $ \binom{S}{k} \coloneqq \brace{ T \subset S : |T| = k } $. Let $[q]=\{1,\ldots,q\}$.

\subsection{List-Decoding}

Fix $ q\in\bbZ_{\ge3} $ and $ L\in\bbZ_{\ge2} $. 
Let $ \disth(\vc,\vr)$ denote the \emph{Hamming distance} between $ \vc,\vr\in[q]^n $, i.e., the number of coordinates on which the strings differ. 
For $ t\in[0,n] $, let $ \ballh(\vy, t) := \{\vc\in [q]^n:\disth(\vc,\vy) \leq t\}$ denote the \emph{Hamming ball} centered around $\vy$ of radius $ \floor{t} $.

\begin{definition}[List-decodable code]
\label{def:listdec}
Let $ p\in[0,1] $. 
A code $ \cC\subseteq [q]^n $ is \emph{$ (p, L)_q $-list-decodable} if for any $ \vy\in[q]^n $, 
\begin{align}
\abs{ \cC \cap \ballh(\vy, np) } &\le L-1 . \notag 
\end{align}
\end{definition}

In~\cite{RYZ22} the zero-rate regime for list-decoding was derived, which is the supremum over $p \in [0,1]$ for which $(p-\eps,L)_q$-list-decodable codes of positive rate exist for all $\eps>0$. This value was shown to be 
\begin{align} \label{eq:list-dec-zero-rate}
    p_*(q,L) = 1 - \frac{1}{L}\exptover{(X_1,\cdots,X_L)\sim U_q^{\ot L}}{\plur(X_1,\cdots,X_L)} ,
\end{align}
where the function $\plur$ outputs the number of times the most popular symbol appears. In~\cite{RYZ22} it is shown that $(p_*(q,L)+\eps,L)$-list-decodable codes have size $O_{\eps,q,L}(1)$, i.e., some constant independent of $n$. Our target in this work is to show that the correct dependence on $\eps$ is $O_{q,L}(1/\eps)$, except for the case of $q=2$ with odd $L$.

A ``dual'' definition of list-decodability is proffered by the Chebyshev radius.

\begin{definition}[Chebyshev radius]
\label{def:cheb_rad}
The \emph{Chebyshev radius} of a list of distinct vectors $ \vc_1, \cdots, \vc_L\in[q]^n $ is defined as
\begin{align}
\radh(\vc_1, \cdots, \vc_L) &\coloneqq 
\frac{1}{n}\min_{\vr\in[q]^n} \max_{i\in[L]} \disth(\vc_i, \vr) . \notag 
\end{align}
\end{definition}

Observe that a code $\cC \subseteq [q]^n$ is $(p,L)$-list-decodable if and only if 
\begin{align}
    \min\{ \radh(\vec c_1,\dots,\vec c_L) : \vec c_1,\dots,\vec c_L \in \cC\text{ distinct}\} > p \ . \label{eq:cheby-condn}
\end{align}
In particular, to show a code fails to be list-decodable, it suffices to upper bound the Chebyshev radius of $L$ distinct codewords. 

Recall that our main target is an upper bound on the size of list-decodable/-recoverable codes (in the zero-rate regime). A natural approach is to derive from \Cref{eq:cheby-condn} the desired bound on the code; however, this quantity is quite difficult to work with directly. We therefore work with a \emph{relaxed} version, which we now introduce.

We require the following definitions. 
Let us embed $ [q]^n $ into the simplex $ \Delta([q]) $ via the following map $\phi$: 
\begin{align}
\begin{array}{rclc}
\phi\colon & [q] &\to & \Delta([q]) \\
 & x &\mapsto & \ve_x
\end{array}
\end{align}
where $ \ve_x $ is the $q$-dimensional vector with a $1$ in the $x$-th location and $0$ everywhere else. 
Denote by $ \Delta = \Delta([q]) $ the simplex and $ \partial\Delta = \{\ve_1, \cdots, \ve_q\} $ its vertices. 
For $ \vchi = \ve_x \in\partial\Delta $ and $ \veta\in\Delta $, let
\begin{align}
d(\vchi, \veta) &:= \frac{1}{2} \normone{\vchi - \veta} = \frac{1}{2} \paren{1 - \veta(x) + \sum_{x'\in[q]\setminus\{x\}} \veta(x')} . \label{eqn:ell1_dist} 
\end{align}
Note that if $ \veta = \ve_y \in\partial\Delta $, then 
\begin{align}
d(\vchi, \veta) &= \disth(x,y) . \notag 
\end{align}

From now on we will only work with $ \Delta^n $-valued vectors and will still denote such length-$qn$ vectors by boldface letters, abusing the notation. 
For $ \vy\in\Delta^n $, we use $ y(j,k)\in[0,1] $ to denote its $(j,k)$-th element and use $ \vy(j)=(y(j,1),\ldots,y(j,q))\in\Delta $ to denote its $j$-th block of size $q$.
For $ \vc\in[q]^n $, we use $ \vc(j)\in[q] $ to denote its $j$-th element. 

For $ \vx\in(\partial\Delta)^n $ and $ \vy\in\Delta^n $, the definition of $ d(\cdot,\cdot) $ can be extended to length-$qn$ vectors in the natural way. 
Specifically, 
\begin{align}
d(\vx, \vy) &= \sum_{j = 1}^n d(\vx(j), \vy(j)) . \label{eqn:ell1_dist_vec} 
\end{align}

We may now define the \emph{relaxed Chebyshev radius}.

\begin{definition} \label{def:def-relaxed-cheb-rad}
    The \emph{relaxed Chebyshev radius} of a list of distinct vectors $ \vx_1, \cdots, \vx_L\in(\partial\Delta)^n $ is
    \begin{align}
        \rad(\vx_1, \cdots, \vx_L) &\coloneqq \frac{1}{n}\min_{\vy\in\Delta^n} \max_{i\in[L]} d(\vx_i, \vy) . \label{eqn:def-relaxed-cheb-rad} 
    \end{align}
\end{definition}

Observe that 
\begin{align}
    \rad(\vx_1, \cdots, \vx_L) &\le \radh(\vc_1, \cdots, \vc_L) . \label{eqn:less} 
\end{align}
where $\phi(\vc_i)=\vx_i$ (here we extend the definition of $\phi$ to length-$n$ inputs in a similar way as in \Cref{eqn:ell1_dist_vec}). This justifies the ``relaxation'' terminology. 

As a last piece of terminology, we define the radius of a code. 

\begin{definition} \label{def:radius-of-code}
    For any code $\cC\in [q]^n$, we define the \emph{Chebyshev radius} of $\cC$ as 
    $$
        \rad(\cC)=\frac{1}{n}\min_{\vx\in [q]^n}\max_{\vc\in \cC} \disth(\vc,\vx).
    $$
\end{definition}

\subsection{List-Recovery}

\begin{table}[htbp]
    \centering
    \begin{tabular}{c|c}
        $\cX=\binom{[q]}{\ell}$ & collection of $\ell$-subsets in $[q]$.\\
        English capital letter $A$ &  a $\ell$-subset in $\cX$ \\
        English capital letter in bold $\vY$ &  $\cX$-valued vector \\
        $\cX_i=\{A\in \cX: i\in A\}$ & collection of $\ell$-subsets containing in $[q]$ that contains element $i$.\\ 
        $\Delta_\ell = \Delta_\ell(\cX)$ & Simplex in $[0,1]^{\cX}$, i,e., $\Delta_\ell=\{(x_A)_{A\in \cX}\in [0,1]^\cX: \sum_{A\in \cX} x_A=1\}$ \\
        $\ve_A$ & the $A$-th vertex of $\Delta_\ell$ \\ $\ve_i=\sum_{A\in \cX_i} \ve_A$ & the image of $i\in [q]$ under $\phi_\ell$\\
        $\partial \Delta_{\ell}=\{\ve_1,\ldots,\ve_q\}$ & The image of elements in $[q]$ under $\phi_\ell$  \\
        $\vc_i\in[q]^n$ & The $i$-th codeword in a list \\
        $\vx_i\in\Delta_\ell^n$ & Image of $\vc_i$ under $\phi$ (applied component-wise) \\
        $\vy\in\Delta_\ell^n$ & Relaxed center of a list \\
        $\vx(j)\in\partial\Delta_\ell, \vy(j)\in\Delta_\ell$ & The $j$-th block (of length $\binom{q}{\ell}$) in $\vx\in(\partial\Delta_\ell)^n, \vy\in\Delta_\ell^n$, respectively \\
        $\vx(j,A)\in\{0,1\}, \vy(j,A)\in[0,1]$  & The $(j,A)$-th element of $\vx\in(\partial\Delta_\ell)^n, \vy\in\Delta_\ell^n$, respectively \\
        $\rad_\ell$ & (Standard) $\ell$ radius \\
        $\rad$ & Relaxed $\ell$-radius \\
        $\ol{\rad}$ & Average $\ell$-radius \\
    $\ol{\rad}_{\omega,\ell}$ & Average $\ell$-radius weighted by $\omega\in\Delta([L])$ \\
        $f_\ell(P, \omega)$ & Expected average $\ell$-radius (weighted by $\omega$) of $P$-distributed symbols \\
        $(X_1, \cdots, X_L)\sim P^{\ot L}$ & A list of i.i.d.\ $P$-distributed symbols \\
        $U_k$ & Uniform distribution on $[k]$
    \end{tabular}
    \caption{Notation for list-recovery. }
    \label{tab:listdec-notation}
\end{table}
We now provide the necessary modifications to the above definitions to the setting of list-recovery. Let $\cX=\binom{[q]}{\ell}$ be the collection of all $\ell$-subsets in $[q]$ and $\cX_i=\{A\in \cX: i\in A\}$. Define $\Delta_\ell=\{(x_A)_{A\in \cX}\in [0,1]^\cX: \sum_{A\in \cX} x_A=1\}$. 
Let $(\ve_A)_{A\in \cX}$ is a standard basis of $\bbR^{\cX}$. 
Let $\ve_i=\sum_{A\in \cX_i}\ve_A$ and $\partial \Delta=\{\ve_1,\ldots,\ve_q\}$. Let $\phi_\ell: [q]\rightarrow \partial \Delta_\ell$ be defined as $\phi_{\ell}(i)=\ve_i$. 
Below, we define the list-recovery distance between a vector in $[q]^n$ and an $n$-tuple of $\ell$-subsets $\vY\in \cX^n$. If $\ell=1$, this list-recovery distance recovers the classic Hamming distance, viewing $\vY$ naturally as a vector in $[q]^n$.
\begin{definition}[List-recovery distance]
\label{def:metric-list-rec}
Given a vector $\vx \in [q]^n$ and a tuple of sets $\vY = (Y_1,\dots,Y_n) \in \cX^n$ for $1 \leq \ell \leq q-1$, we define
\begin{align}
    \distlr(\vx,\vY) &\coloneqq \sum_{i=1}^n \indicator{x_i \notin Y_i} \ .
    \notag
\end{align} 
\end{definition}

\begin{definition}[List-recoverability]
\label{def:list-rec}
A code $ \cC\subset[q]^n $ is said to be \emph{$ (p,\ell,L)$-list-recoverable} if for every $\vY \in \binom{[q]}{\ell}^n$, $\card{\cC \cap \blr(\vY,np)} < L$ where 
$$
\blr(\vY,np)=\{\vc\in [q]^n: \distlr(\vc,\vY)\leq np\}.
$$
\end{definition}

The zero-rate regime for list-recoverability was also derived in~\cite{RYZ22}:
\begin{align} \label{eq:list-rec-zero-rate}
    p_*(q,\ell,L) = 1 - \frac{1}{L}\exptover{(X_1,\cdots,X_L)\sim U_q^{\ot L}}{\plur_\ell(X_1,\cdots,X_L)} \ ,
\end{align}
where $\plur_\ell(x_1,\dots,x_L) = \max_{\Sigma \subseteq [q]:|\Sigma|=\ell} |\{i \in [L]:x_i \in \Sigma\}|$ is the top-$\ell$ plurality value, i.e., the number of times the $\ell$ most popular symbols appear. 

\begin{definition}[$\ell$-radius]
\label{def:list-rec-rad}
The \emph{$\ell$-radius} of an $L$-set of vectors $ \vc_1,\dots,\vc_L\in[q]^n $ is defined as the radius of the smallest list-recovery ball containing the set $\{\vc_1,\dots,\vc_L\}$:
\begin{align}
    \rad_\ell(\vc_1,\dots,\vc_L) &\coloneqq \frac{1}{n}\min_{\vY\in\cX^n} \max_{i\in[L]} \distlr(\vc_i,\vY) . \label{eqn:def-rad-lr} 
\end{align}
\end{definition}
In analogy to the list-decoding case, we define the $\ell$-radius of $\cC$. 
\begin{definition}\label{def:lradius-of-code}
 For any code $\cC\in [q]^n$, we define the \emph{$\ell$-radius} of $\cC$ as
$$
\rad_\ell(\cC)=\frac{1}{n}\min_{\vY\in \cX^n}\max_{\vc\in \cC} \distlr(\vc,\vY).
$$
\end{definition}

We embed $ [q]$ into Euclidean space $[0,1]^{\cX}$,

\begin{align}
\begin{array}{rclc}
\phi_\ell\colon & [q] &\to & [0,1]^{\cX} \\
 & i &\mapsto & \ve_{i}:=\sum_{A\in \cX_i} \ve_A.
\end{array}
\end{align}
Define $\partial \Delta_\ell=\{\ve_1,\ldots,\ve_q\}$.
For $ \vchi = \ve_{i} \in\partial\Delta_\ell $ and $ \veta\in\Delta_\ell $, let
\begin{align}
d(\vchi, \veta) &= \frac{1}{2} \paren{ \normone{\vchi - \veta}-\binom{q-1}{\ell-1}+1 } \notag \\
&= \frac{1}{2} \paren{\sum_{A\in \cX_i}(1 - \veta(A)) + \sum_{A'\in \cX\setminus \cX_i} \veta(A')-\binom{q-1}{\ell-1}+1}. \label{eqn:ell2_dist} 
\end{align}
We abuse the notation $d(\vchi, \veta)$ as $\vchi, \veta$ are vectors of length $|\cX|$. 
Note that if $ \veta = \ve_A$ for some $A\in \cX$, then 
\begin{align}
d(\vchi, \veta) &= \distlr(i,A)\in \{0,1\}. \notag 
\end{align}
Define the \emph{relaxed $\ell$-radius} of a list of distinct vectors $ \vx_1, \cdots, \vx_L\in(\partial\Delta_\ell)^n $ as:
\begin{align}
\rad(\vx_1, \cdots, \vx_L) &\coloneqq \frac{1}{n}\min_{\vy\in\Delta_\ell^n} \max_{i\in[L]} d(\vx_i, \vy) . \notag 
\end{align}
Obviously, 
\begin{align}
\rad(\vx_1, \cdots, \vx_L) &\le \rad_{\ell}(\vc_1, \cdots, \vc_L) . \notag 
\end{align}
where $\vx_i=\phi_{\ell}(\vc_i)$.

\subsection{Types of Vector Tuples}

The last concept that we need is the type of a tuple of vectors.  Informally, one takes a tuple of vectors $(\vc_1,\dots,\vc_L) \in ([q]^n)^L$, views it as a $L \times n$ matrix, and then computes the fraction of columns that take on a certain value $\vu \in [q]^L$ for each $\vu \in [q]^L$. In other words, the type of is the distribution on $[q]^L$ induced by randomly sampling a column from this matrix.

\begin{definition}[Type]
\label{def:type}
    Let $q, L \in \bbN$, let $(\vc_1,\dots,\vc_L) \in ([q]^n)^L$ be a tuple of vectors, and let $\vu \in [q]^L$ be a vector. The \emph{type} of $(\vc_1,\dots,\vc_L)$ is 
    \[
        \typ(\vc_1,\dots,\vc_L) \coloneqq (\typ_\vu(\vc_1,\dots,\vc_L))_{\vu \in [q]^L}
    \]
    where 
    \[
        \typ_\vu(\vc_1,\dots,\vc_L) \coloneqq \frac{1}{n}\sum_{i=1}^n\indicator{(\vc_{1}(i),\dots,\vc_{L}(i)) = \vu} \ .
    \]
\end{definition}



%% file: list-dec-radius.tex
We have shown in \Cref{eqn:less} that $\rad$ is smaller than $\radh$. 
Conversely, \Cref{lm:convert} below establishes that $ \rad $ and $ \radh $ do not differ much. 
That is, for any list, if a center $ \vy\in\Delta^n $ achieves a relaxed radius $t$, then there must exist $ \vr\in[q]^n $ attaining approximately the same $t$ for sufficiently large $n$. 

\begin{lemma}[$\rad$ is close to $ \radh $]\label{lm:convert}
Let $ \vc_1, \cdots, \vc_L \in [q]^n $. 
Denote by $ \vx_1, \cdots, \vx_L\in(\partial\Delta)^n $ the images of $ \vc_1, \cdots, \vc_L $ under the embedding $\phi$. 
Then 
\begin{align}
\radh(\vc_1, \cdots, \vc_L)
&\le \rad(\vx_1, \cdots, \vx_L) + \frac{L}{n} . \notag 
\end{align}
\end{lemma}

\begin{proof}
Suppose $ \rad(\vx_1, \cdots, \vx_L) = t $. 
Then there exists $ \vy\in\Delta^n $ such that for every $ i\in[L] $, 
\begin{align}
d(\vx_i, \vy) &= \frac{1}{2} \sum_{j = 1}^n \paren{ 1 - y(j, c_i(j)) + \sum_{x\in[q]\setminus\{c_i(j)\}} y(j, x) }
\le t , \notag 
\end{align}
where the first equality is by \Cref{eqn:ell1_dist,eqn:ell1_dist_vec}. 
That is, the following polytope is nonempty:
\begin{multline}
\brace{
    \vy \in \Delta^n : 
    \forall i\in[L], \; d(\vx_i, \vy) \le t
} \\
= \brace{
    (y(j,k))_{(j, k)\in[n]\times[q]} : \begin{array}{l}
    \forall (j, k)\in[n]\times[q],\, y(j,k) \ge 0 , \\
    \forall j\in[n], \, \sum_{k = 1}^q y(j,k) = 1 , \\ 
    \forall i\in[L], \, \frac{1}{2} \sum_{j = 1}^n \paren{ 1 - y(j, \vc_i(j)) + \sum_{x\in[q]\setminus\{\vc_i(j)\}} y(j,x) } \le t
    \end{array}
} . \notag 
\end{multline}
Equivalently, the following linear program (LP) is feasible: 
\begin{align}
\begin{array}{cl}
\max\limits_{(y(j,k))_{(j, k)\in[n]\times[q]}} & 0 \\
\mathrm{s.t.} & \forall (j, k)\in[n]\times[q],\, y(j,k) \ge 0 , \\
    & \forall j\in[n], \, \sum_{k = 1}^q y(j,k) = 1 , \\ 
    & \forall i\in[L], \, \frac{1}{2} \sum_{j = 1}^n \paren{ 1 - y(j, \vc_i(j)) + \sum_{x\in[q]\setminus\{\vc_i(j)\}} y(j,x) } \le t . 
\end{array}
\notag 
\end{align}
Since the equality $ \inprod{\va}{\vy} \leq b $ is equivalent to the equality $ \inprod{\va}{\vy} + z = b , z\geq 0$, the above LP can be written in equational form: 
\begin{align}
\begin{array}{cl}
\max\limits_{(y(j,k))_{(j, k)\in[n]\times[q]}, (z(i))_{i\in[L]}} & 0 \\
\mathrm{s.t.} & \forall (j, k)\in[n]\times[q],\, y(j,k) \ge 0 , \\
    & \forall i\in[L], \, z(i) \ge 0 , \\
    & \forall j\in[n], \, \sum\limits_{k = 1}^q y(j,k) = 1 , \\ 
    & \forall i\in[L], \, \frac{1}{2} \sum\limits_{j = 1}^n \paren{ 1 - y(j, \vc_i(j)) + \sum\limits_{x\in[q]\setminus\{\vc_i(j)\}} y(j,x) } + z(i) \le t , 
\end{array}
\label{eqn:LP_eqn_form} 
\end{align}
or more compactly in matrix form: 
\begin{align}
\begin{array}{cl}
\max\limits_{\vy\in\bbR^{nq}, \vz\in\bbR^L} & 0 \\
\mathrm{s.t.} &  \matrix{A & I_L \\ B & 0} \matrix{\vy \\ \vz} = \matrix{t \one_L \\ \one_n} , \\
    & \vy,\vz \ge \zero . 
\end{array}
\notag 
\end{align}
Here $ A\in\bbR^{L \times (nq)} $ and $ B\in\bbR^{n \times (nq)} $ encode respectively the fourth and third constraints in \Cref{eqn:LP_eqn_form},  and $ I_L\in\bbR^{L\times L}, \one_L \in \bbR^L $ denote respectively the $L\times L$ identity matrix and the all-one vector of length $L$. 
It is clear that 
\begin{align}
    \rk\paren{ \matrix{A & I_L \\ B & 0} } &\le n+L . \notag 
\end{align}
This implies that there exists a feasible solution $\vy,\vz$ that has at most $n+L$ nonzeros and thus  $\vy=(y(j,k))_{(j, k)\in[n]\times[q]}$ has at most $n+L$ nonzeros. 
Indeed, such solutions are known as the \emph{basic feasible solutions}; see \Cref{prop:basic_feasible_sol}. 
Note that for every block $j\in [n]$, $\sum_{k=1}^q y(j,k)=1$. This implies that $y(j,1),\ldots,y(j,q)$ cannot be simultaneously $0$. Moreover, if $q-1$ out of them are $0$, the remaining one is forced to be $1$. Since there are $n$ blocks in total, by the pigeonhole principle, there are at least $n-L$ choices of $j\in [n]$ such that $\vy(j) =(y(j,1),\ldots,y(j,q))\in \partial\Delta$. Without loss of generality, we assume that these $n-L$ indices are $1,\ldots,n-L$. Let $\vr\in [q]^n$ be such that $\phi(\vr(j))=\vy(j)$ for $j=1,\ldots,n-L$ and $r(j)$ is any value in $[q]$ for $j=n-L+1,\ldots,n$. Since $d(\vx_i(j),\vy(j))\in [0,1]$ and $\disth(\vc_i(j),\vr(j))\in \{0,1\}$,
 the difference between $d(\vx_i,\vy)$ and $\disth(\vc_i,\vr)$ is at most $L$. The proof is completed.
\end{proof}


We further relax $ \rad $ by defining the \emph{weighted average radius}. 
For $ \vx_1, \cdots, \vx_L\in(\partial\Delta)^n $ and $ \omega\in\Delta([L]) $, let 
\begin{align}
\ol{\rad}_\omega(\vx_1, \cdots, \vx_L) &\coloneqq \frac{1}{n}\min_{\vy\in\Delta^n} \exptover{i\sim \omega}{d(\vx_i, \vy)}
= \frac{1}{n}\min_{\vy\in\Delta^n} \sum_{i\in[L]} \omega(i) d(\vx_i, \vy) . \notag 
\end{align}
In words, weighted average radius is obtained by replacing the maximization over $ i\in[L] $ in the definition of relaxed radius (see \Cref{eqn:def-relaxed-cheb-rad}) with an average with respect to a distribution $\omega$. 

Since the objective of the minimization is separable, one can minimize over each $ y(j) $ individually and obtain an alternative expression. 
Suppose $ \vx_1, \cdots, \vx_L\in(\partial\Delta)^n $ are the images of $ \vc_1, \cdots, \vc_L\in[q]^n $ under the embedding $\phi$. 
Then
\begin{align}
\ol{\rad}_\omega(\vx_1, \cdots, \vx_L)
&= \frac{1}{n}\min_{\vy\in\Delta^n} \sum_{i\in[L]} \omega(i) d(\vx_i, \vy) \notag \\
&= \frac{1}{2n} \min_{(\vy_1, \cdots, \vy_n)\in\Delta^n} \sum_{i\in[L]} \omega(i) \sum_{j = 1}^n \paren{ 1 - \vy_j(\vc_i(j)) + \sum_{x\in[q]\setminus\{\vc_i(j)\}} \vy_j(x) } \notag \\
&= \frac{1}{2n} \min_{(\vy_1, \cdots, \vy_n)\in\Delta^n} \sum_{j = 1}^n \brack{ \sum_{i\in[L]} \omega(i) \paren{ 1 - \vy_j(\vc_i(j)) + \sum_{x\in[q]\setminus\{\vc_i(j)\}} \vy_j(x) } } \notag \\
&= \frac{1}{2n} \sum_{j = 1}^n \min_{\vy_j\in\Delta} \brack{ \sum_{i\in[L]} \omega(i) \paren{ 1 - 2\vy_j(\vc_i(j)) + \sum_{x\in[q]} \vy_j(x) } } \label{eqn:wt_avg_rad_1} \\
&= \frac{1}{2n} \sum_{j = 1}^n \min_{\vy_j\in\Delta} \brack{ \sum_{i\in[L]} \omega(i) \paren{ 2 - 2\vy_j(\vc_i(j)) } } \notag \\
&= \frac{1}{n}\sum_{j = 1}^n \min_{\vy_j\in\Delta} \brack{ 1 - \sum_{i\in[L]} \omega(i) \vy_j(\vc_i(j)) } \notag \\
&= 1 - \frac{1}{n}\sum_{j = 1}^n \max_{\vy_j\in\Delta} \brack{ \sum_{i\in[L]} \omega(i) \vy_j(\vc_i(j)) } \notag \\
&= 1 - \frac{1}{n}\sum_{j = 1}^n \max_{x\in[q]} \sum_{\substack{i\in[L] \\ \vc_i(j) = x}} \omega(i) . \label{eqn:wt_avg_rad_2} 
\end{align}
\Cref{eqn:wt_avg_rad_1} holds since the objective in brackets only depends on $ \vy_j $, not on other $ (\vy_{j'})_{j'\in[n]\setminus\{j\}} $. 
To see \Cref{eqn:wt_avg_rad_2}, we note that a maximizer $ \vy^*\in\Delta $ to the following problem 
\begin{align}
\max_{\vy\in\Delta} \sum_{i\in[L]} \omega(i) y(x_i) , \notag 
\end{align}
where $ \omega\in\Delta([L]) $ and $ (x_1, \cdots, x_L)\in[q]^L $ are fixed, is given by $ \vy^* = \ve_{x^*} $ where $ x^*\in[q] $ satisfies 
\begin{align}
    x^* &\in \argmax_{x\in[q]} \sum_{i\in[L]} \omega(i) \indicator{x_i = x} . \notag 
\end{align}

Obviously, by definition, for any $ \vx_1, \cdots, \vx_L \in (\partial\Delta)^n $ and $ \omega\in\Delta([L]) $, 
\begin{align}
\ol{\rad}_\omega(\vx_1, \cdots, \vx_L) &\le \rad(\vx_1, \cdots, \vx_L) . \notag 
\end{align}
In fact, the following lemma shows that $\rad$ is equal to the maximum $ \ol{\rad}_\omega $ over $\omega$. 

\begin{lemma}[$ \rad $ equals maximum $ \ol{\rad}_\omega $]
For any $ \vx_1, \cdots, \vx_L\in(\partial\Delta)^n $, 
\begin{align}
\rad(\vx_1, \cdots, \vx_L) &= \max_{\omega\in\Delta([L])} \ol{\rad}_\omega(\vx_1, \cdots, \vx_L) . \notag 
\end{align}
\end{lemma}

\begin{proof}
Note that 
\begin{align}
\rad(\vx_1, \cdots, \vx_L) &\coloneqq \frac{1}{n}\min_{\vy\in\Delta^n} \max_{i\in[L]} d(\vx_i, \vy) 
= \frac{1}{n}\min_{\vy\in\Delta^n} \max_{\omega\in\Delta([L])} \sum_{i\in[L]} \omega(i) d(\vx_i, \vy) , \notag 
\end{align}
since the inner maximum is anyway achieved by a singleton distribution. 
Note also that the objective function 
\begin{align}
\sum_{i\in[L]} \omega(i) d(\vx_i, \vy)
&= \frac{1}{2} \sum_{i\in[L]} \omega(i) \sum_{j = 1}^n \paren{ 1 - y(j, \vc_i(j)) + \sum_{x\in[q]\setminus\{\vc_i(j)\}} y(j, x) } \notag 
\end{align}
is affine in $\omega$ and linear in $ \vy $. 
Therefore, von Neumann's minimax theorem allows us to interchange $\min$ and $\max$ and obtain
\begin{align}
\rad(\vx_1, \cdots, \vx_L) &= \frac{1}{n}\max_{\omega\in\Delta([L])} \min_{\vy\in\Delta^n} \sum_{i\in[L]} \omega(i) d(\vx_i, \vy) 
= \max_{\omega\in\Delta([L])} \ol{\rad}_\omega(\vx_1, \cdots, \vx_L) , \notag 
\end{align}
as claimed by the lemma. 
\end{proof}


In fact, we can say something stronger: it is not necessary to maximize over the entire (uncountable) probability simplex $\Delta([L])$. Instead, we can extract a finite subset $\Omega_L \subset \Delta([L])$ and maximize over this set to recover $\rad$. The following lemma is analogous to \cite[Lemma~6]{abp-2018}.


\begin{lemma}[$\rad$ is achieved by finitely many $ \omega $]\label{lm:finite}
For every $L$, there exists a finite set of probability measures $\Omega_L\subseteq \Delta([L])$ such that 
$$
\rad(\vx_1,\ldots,\vx_L)=\max_{\omega\in \Omega_L}\ol{\rad}_{\omega}(\vx_1,\ldots,\vx_L).
$$
for all $\vx_1,\ldots,\vx_L\in \partial \Delta^n$. 
\end{lemma}
\begin{proof}
The idea is to view the computation of $\max_{\omega \in \Omega_L}\ol{\rad}_{\omega}(\vx_1,\dots,\vx_L)$ as finding the maximum among some finite set of linear program maxima over some convex polytopes, and then to take $\Omega_L$ to be the set of vertices of the defined convex polytopes. 

First, we define the convex polytopes based on a ($q$-ary version of a) signature. For each $\omega\in \Delta([L])$, we define a \emph{signature} for $\omega$ which is a function $S_\omega: [q]^L\rightarrow [q]$ such that
\[
    S_\omega(\vu) \in \argmax_{x\in [q]}\sum_{i:\vu(i)=x}\omega(i)
\]
for $\vu\in [q]^L$. Define further the $q$ halfspaces $H_{\vu,x}:=\{\omega \in \Delta([L]): \sum_{i:\vu(i)=x}\omega(i) \geq 1/q\}$ for $x \in [q]$. 
Observe that if $S(\vu)=x$ where $S$ is a signature for $\omega$ then $\omega \in H_{\vu,x}$. Thus, by ranging over the choices for $\vu \in [q]^L$ and $x \in [q]$ we obtain $q^{L+1}$ halfspaces that partition the $(L-1)$-dimensional space $\Delta([L])$ into at most $\sum_{j\leq L-1}\binom{q^{L+1}}{j}$ regions.

For each possible signature $S:[q]^L \to [q]$, let $\Omega_S=\{\omega\in \Delta([L])\colon S \text{ is a signature for } \omega\}$, and note that $\Omega_S$ is a convex polytope. Indeed, it is an intersection over $\vu \in [q]^L$ of the convex polytopes 
\begin{align*}
    \left\{\omega \in \Delta([L])\colon \exists \text{ signature } \right.&\left.S_\omega \text{ for } \omega \text{ s.t. }S_\omega(\vu) = S(\vu)\right\} \\
    &= \bigcap_{y \in [q]\setminus S(\vu)} \left\{\omega \in \Delta([L])\colon \sum_{i:\vu(i)=S(\vu)}\omega(i) \geq \sum_{i:\vu(i)=y}\omega(i)\right\}
\end{align*}
where $S_\omega$ is a signature for $\omega$. 
Now, to maximize 
\[
    \ol{\rad}_\omega(\vx_1, \cdots, \vx_L)
= \frac{1}{n}\min_{\vy\in\Delta^n} \sum_{i\in[L]} \omega(i) d(\vx_i, \vy) 
\]
over $\omega \in \Omega_L$, consider the set $T_\vu=\{i\in [n]: (\vx_1(i),\ldots,\vx_L(i))=\vu\}$ for $\vu\in [q]^L$ and let $a_{\vu}=\frac{|T_\vu|}{n}$. We claim it suffices to find the maximum of the following linear function: 
\begin{align} \label{eq:to-maximize-relaxed}
\sum_{\vu\in [q]^L}a_{\vu}y_{\vu},\quad \text{s.t.}\quad y_{\vu}=\sum_{i:\vu(i)=S(\vu)}\omega(i)
\end{align}
over all $\omega\in \Omega_S$. Indeed, by \Cref{eqn:wt_avg_rad_2}, we have 
$$
\ol{\rad}_\omega(\vx_1, \cdots, \vx_L)
= 1 - \frac{1}{n}\sum_{j = 1}^n \max_{x\in[q]} \sum_{\substack{i\in[L] \\ c_i(j) = x}} \omega(i).
$$
This implies that a maximizer only depends on the index set $T_\vu$, and furthermore that its value is determined by the $a_{\vu}$'s as in \Cref{eq:to-maximize-relaxed}.


We can thus take the union of all vertex sets of all polytopes $\Omega_S$ for all signatures $S$.
Multiplying this by the $O_{q,L}(1)$ regions defined by all the halfspaces $H_{\vu,x}$ we obtain a finite set of vertices, as desired. 
\end{proof}

%% file: list-dec-function-props.tex
Now, we consider the expected weighted average radius of a sequence of i.i.d.\ symbols. 
Specifically, for $ P\in\Delta([q]) $ and $ \omega\in\Delta([L]) $, let
\begin{align}
f(P, \omega) &\coloneqq \exptover{(X_1, \cdots, X_L)\sim P^{\ot L}}{\ol{\rad}_\omega(\ve_{X_1}, \cdots, \ve_{X_L})} . \notag 
\end{align}
\cite{abp-2018} studies $f(P,\omega)$ for $q=2$ and even $L$. In this case, one can take advantage of the fact that $P \in \Delta([2])$ may be parametrized by a single real number, and thereby yield a fairly simple expression for $f(P,\omega)$. 

Nonetheless, in this subsection, we will show that all properties of $f(P,\omega)$ in \cite{abp-2018} holding for $q=2$ and even $L$ can be generalized to any $q\geq 3$ and any $L$. Let us first provide a more explicit expression for $ f(P, \omega) $ using \Cref{eqn:wt_avg_rad_2}: 
\begin{align}
f(P, \omega) &\coloneqq \exptover{(X_1, \cdots, X_L)\sim P^{\ot L}}{1 - \max_{x\in[q]} \sum_{\substack{i\in[L] \\ X_i = x}} \omega(i)} \notag \\
&= 1 - \sum_{(x_1, \cdots, x_L)\in[q]^L} \paren{\prod_{i = 1}^L P(x_i)} \max_{x\in[q]} \sum_{i\in[L]} \omega(i) \indicator{x_i = x} . \notag
\end{align}

We define the shorhand notation 
\begin{align} \label{eq:max_omega-def}
    \mathrm{max}_{\omega}(x_1, \cdots, x_L) &\coloneqq \max_{x\in[q]} \sum_{\substack{i\in[L] \\ x_i = x}} \omega(i)
\end{align}
for any $\omega \in \Delta([L])$ and $(x_1,\dots,x_L) \in [q]^L$.

The first property that we would like to establish is that $f(P,\omega)$ only increases if $\omega$ is replaced by $U_L$, and furthermore that the maximum is uniquely obtained at $U_L$ if $P(x)>0$ for all $x \in [q]$. In order to do this, we will regularly ``average-out'' coordinates of $\omega$ and then show that the function value increases (or at least, does not decrease). To be introduce some terminology, for $S \subseteq [L]$ we say that $\ol \omega$ is obtained from $\omega$ by \emph{averaging-out} the subset $S$ of coordinates if $\ol \omega$ is defined as
\[
    \ol\omega(i) = \begin{cases}
        \frac{\sum_{j \in S}\omega(j)}{|S|} & i \in S \\
        \omega(i) & i \notin S
    \end{cases} \ .
\]
The following lemma gives a simple criterion for establishing that, if $\ol\omega$ is obtained from $\omega$ by averaging two coordinates, then $f(P,\ol\omega) \leq f(P,\omega)$, and it furthermore gives a criterion for the inequality to be strict. The main thrust of the proof of \Cref{lm:convexity} is thus to show that this criterion is always satisfied. 


\begin{lemma} \label{lem:criterion-for-increase}
    Let $P \in \Delta([q])$ and $\omega \in \Delta([L])$. Suppose $\omega(L-1)\neq \omega(L)$ and that $\ol\omega \in \Delta([L])$ is obtained by averaging-out the last two coordinates of $\omega$. Suppose that for all $(x_1,\dots,x_L) \in [q]^L$ we have 
    \begin{align}
        \frac{1}{2} \paren{ \mathrm{max}_{\omega}(x_1, \cdots, x_{L-1}, x_L) + \mathrm{max}_{\omega}(x_1, \cdots, x_{L}, x_{L-1}) }
        &\ge \mathrm{max}_{\ol{\omega}}(x_1, \cdots,x_{L-1}, x_L) . \label{eqn:maxq_convex} 
    \end{align}
    Then $f(P,\ol\omega) \geq f(P,\omega)$. 
    
    Furthermore, suppose that additionally there exists $(x_1,\dots,x_L) \in [q]^L$ with $\prod_{i=1}^LP(x_i)>0$ such that the inequality in \Cref{eqn:maxq_convex} is strict. Then $f(P,\ol \omega) > f(P,\omega)$. 
\end{lemma}

\begin{proof}
    Define $ \omega'\in\Delta([L]) $ as 
\begin{align}
\omega'(i) &= \begin{cases}
\omega(i) , & i\in[L]\setminus\{L-1,L\} \\
\omega(L) , & i=L-1 \\
\omega(L-1) , & i=L
\end{cases} . \notag
\end{align}
That is, $\omega'$ is obtained by swapping the last two components of $\omega$. By symmetry, we have $f(P,\omega)=f(P,\omega')$ and so
\begin{align*}
    f(P,\omega) &= \frac{1}{2}(f(P,\omega)+f(P,\omega'))\\
    &= \frac{1}{2}\left(1-\sum_{(x_1, \cdots, x_L)\in[q]^L}\paren{\prod_{i=1}^LP(x_i)}\mathrm{max}_\omega(x_1,\ldots,x_{L-1},x_{L}) \right.\\
    &\quad\quad\quad\left.+ 1-\sum_{(x_1, \cdots, x_L)\in[q]^L}\paren{\prod_{i=1}^LP(x_i)}\mathrm{max}_{\omega'}(x_1,\ldots,x_{L-1},x_{L})\right)\\
    &= 1-\sum_{(x_1, \cdots, x_L)\in[q]^L}  \paren{\prod_{i=1}^LP(x_i)}\frac{1}{2} \left(\mathrm{max}_\omega(x_1,\ldots,x_{L-1},x_{L})+\mathrm{max}_\omega(x_1,\ldots,x_L,x_{L-1}))\right) \\ 
    &\leq 1-\sum_{(x_1, \cdots, x_L)\in[q]^L}  \paren{\prod_{i=1}^LP(x_i)}\mathrm{max}_{\ol\omega}(x_1,\ldots,x_{L-1},x_{L})\\
    &= f(P,\ol\omega) \ ,
\end{align*}
where the inequality follows from \Cref{eqn:maxq_convex}. From the above sequence of inequalities, it is also clear that if additionally there exists $(x_1,\dots,x_L) \in [q]^L$ with $\prod_{i=1}^LP(x_i)>0$ for which the inequality in \Cref{eqn:maxq_convex} is strict, then $f(P,\ol \omega) > f(P,\omega)$. 
\end{proof}

We now establish that the function value cannot decrease if $\omega$ is replaced by $U_L$. 

\begin{lemma}\label{lm:convexity}
\leavevmode
 For any $ P\in\Delta([q]) $ and $ \omega\in\Delta([L]) $, $ f(P, \omega) \le f(P, U_L) $. 
\end{lemma}

\begin{proof}
Fix any $ (x_1, \cdots, x_L)\in[q]^L $. 
Let $ \omega\in\Delta([L]) $ be non-uniform. 
Without loss of generality, assume $ \omega(L-1)\neq \omega(L) $. 
Let $ \ol{\omega}\in\Delta([L]) $ be obtained by uniformizing the last two components of $\omega$, i.e., 
\begin{align}
\ol{\omega}(i) &= \begin{cases}
\omega(i) , & i\in[L]\setminus\{L-1,L\} \\
\frac{1}{2} (\omega(L-1) + \omega(L)) , & i\in\{L-1,L\}
\end{cases} . \notag 
\end{align}
We claim $f(P,\ol \omega) \geq f(P,\omega)$. By \Cref{lem:criterion-for-increase}, we just need to establish \Cref{eqn:maxq_convex}. 

\Cref{eqn:maxq_convex} trivially holds if $ x_{L-1} = x_L $. 
We therefore assume below $ x_{L-1}\ne x_L $. Let $x_{L-1}=a$ and $x_L=b$. 
Let 
$$
\omega^{(a)}=\sum_{\substack{i\in[L-2] \\ x_i = a}}\omega(i),\qquad \omega^{(b)}=\sum_{\substack{i\in[L-2] \\ x_i = b}}\omega(i) \ .
$$
Then, we have 
$$
\sum_{\substack{i\in[L] \\ x_i = a}} \ol{\omega}(i)=\omega^{(a)}+\frac{1}{2}(\omega(L-1)+\omega(L)), \qquad \sum_{\substack{i\in[L] \\ x_i = b}} \ol{\omega}(i)=\omega^{(b)}+\frac{1}{2}(\omega(L-1)+\omega(L)).
$$
We first assume that there exists $c\notin \{a,b\}$ such that 
$$
\mathrm{max}_{\ol{\omega}}(x_1, \cdots, x_L)=\sum_{\substack{i\in[L] \\ x_i = c}}\ol{\omega}(i)=\sum_{\substack{i\in[L] \\ x_i = c}}{\omega}(i)
$$
where the second equality follows since the set $\{i\in [q]: x_i=c\}$ does not contain $L-1,L$. \Cref{eqn:maxq_convex} therefore holds as
$$
\mathrm{max}_{{\omega}}(x_1, \cdots,x_{L-1}, x_L)\geq \sum_{\substack{i\in[L] \\ x_i = c}}{\omega}(i), \qquad
\mathrm{max}_{\omega}(x_1, \cdots, x_L,x_{L-1})\geq \sum_{\substack{i\in[L] \\ x_i = c}}{\omega}(i).
$$
We proceed to the case that 
$$
\mathrm{max}_{\ol{\omega}}(x_1, \cdots, x_L)=\max\left\{\frac{1}{2}(\omega(L-1)+\omega(L))+\omega^{(a)},\frac{1}{2}(\omega(L-1)+\omega(L))+\omega^{(b)}\right\}.
$$
\Cref{eqn:maxq_convex} holds as
$$
\mathrm{max}_{\omega}(x_1, \cdots,x_{L-1}, x_L)\geq \max\left\{ \omega^{(a)}+\omega(L-1),\omega^{(b)}+\omega(L) \right\}
$$
and
$$
\mathrm{max}_{{\omega}}(x_1, \cdots,x_{L}, x_{L-1})\geq \max\left\{ \omega^{(a)}+\omega(L),\omega^{(b)}+\omega(L-1) \right\}.
$$

Thus, \Cref{lem:criterion-for-increase} implies $f(P,\ol \omega) \geq f(P,\omega)$, as desired. 

We can then continue averaging components of $\omega$ and in this way obtain a sequence $(\omega_i)_{i \in \bbN}$ of distributions with $\omega_1 = \omega$. This sequence converges in $\ell_\infty$-norm to the uniform distribution $U_L$ and satisfies $f(P,\omega_{i+1}) \geq f(P,\omega_i)$ for all $i \in \bbN$. Observing that $\omega \mapsto f(P,\omega)$ is a continuous function -- the term $\max_{x \in [q]} \sum_{i \in [L]} \omega(i) \indicator{x_i = x}$ is a maximum over linear functions of $\omega$, hence linear, implying that $f(P,\cdot)$ is a linear combination of continuous functions -- it follows that $f(P,U_L) = \lim_{i \to \infty} f(P,\omega_i)$, and in particular that $f(P,U_L) \geq f(P,\omega_1) = f(P,\omega)$, as desired. 
\end{proof}


We now strengthen the conclusion of \Cref{lm:convexity} by showing that for all $q \geq 3$ and $L \geq 2$ the function $\omega \mapsto f(P,\omega)$ is \emph{uniquely} maximized by the setting $\omega = U_L$, except for degenerate cases where $P(x)=0$ for some $x \in [q]$. 

Before stating and proving this fact, we note that the proof of \Cref{lm:convexity} in fact shows that we can average out any subset of coordinates of $\omega$ and only increase the value of $f(P,\omega)$. We formalize this fact in the following lemma, which will be useful in the following arguments. 

\begin{lemma} \label{lem:average-subset-omega}
    Let $P \in \Delta([q])$, $\omega \in \Delta([L])$ and $S \subseteq [L]$. Let $\ol\omega$ be obtained from $\omega$ by averaging-out the subset $S$ of coordinates. 
    Then $f(P,\ol\omega) \geq f(P,\omega)$.
\end{lemma}

\begin{theorem} \label{thm:maxq-unique}
    Let $q \geq 3$, $L \geq 2$ and let $P \in \Delta([q])$ be such that $P(x)>0$ \footnote{In fact, our proof only apply with $P=U_q$ which clearly satisfies the condition.} for all $x \in [q]$. Then for all $\omega \in \Delta([L])$, $f(P,\omega) \leq f(P,U_L)$ with equality if and only if $\omega = U_L$.
\end{theorem}

\begin{proof}
    The inequality was already established in \Cref{lm:convexity}, so we focus on showing $\omega=U_L$ when $f(P,\omega) = f(P,U_L)$. As $q \geq 3$, let $a,b$ and $c$ denote $3$ distinct elements of $[q]$. Let $\omega \neq U_L$ and suppose for a contradiction that $f(P,\omega)$ is a maximum of the function $\omega \mapsto f(P,\omega)$. The proof proceeds via a number of cases. 

    \begin{enumerate}
        \item \textbf{$L$ is even.} Without loss of generality, $\omega(L-1) < \omega(L)$. If $L\geq 4$, let $\omega'$ be obtained from $\omega$ by averaging-out the first $L-2$ coordinates; by \Cref{lem:average-subset-omega}, $f(P,\omega') \geq f(P,\omega)$. If $L=2$, set $\omega'=\omega$. 
        
        If $L\geq 4$, since $2|(L-2)$, we can set $x_1=\dots=x_{L/2-1}=a$ and $x_{L/2}=\dots=x_{L-2}=b$. Set further $x_{L-1}=a$ and $x_L=b$. We observe that for this $(x_1,\dots,x_L)\in [q]^L$ and $\ol\omega$ obtained from $\omega'$ by averaging-out the last two coordinates, \Cref{eqn:maxq_convex} strictly holds. Indeed, 
        $$
            \mathrm{max}_{\omega'}(x_1,\ldots,x_{L-1},x_{L})+\mathrm{max}_{\omega'}(x_1,\ldots,x_L,x_{L-1})=2\sum_{j=1}^{(L-2)/2}\omega(i)+2\omega(L)
        $$
        and 
        $$
            2\mathrm{max}_{\ol{\omega}}(x_1,\ldots,x_{L-1},x_{L})=2\sum_{j=1}^{(L-2)/2}\omega(i)+\omega(L-1)+\omega(L).
        $$
        Thus \Cref{lem:criterion-for-increase} implies $f(P,\ol\omega) > f(P,\omega')\geq f(P,\omega)$, a contradiction. 

        \item \textbf{$L$ is odd and at least three components of $\omega$ take distinct values, or the components in $\omega$ only take two different values and at least two of them take the minimum value.} Without loss of generality $\omega(L-2) \leq \omega(L-1) < \omega(L)$. If $L \geq 5$, let $\omega'$ be obtained from $\omega$ by averaging-out the first $L-3$ coordinates; by \Cref{lem:average-subset-omega}, $f(P,\omega') \geq f(P,\omega)$. If $L=3$, set $\omega'=\omega$. 
        
        Since $2|(L-3)$, if $L \geq 5$, we set $x_1=\dots=x_{(L-1)/2-1}=a$ and $x_{(L-1)/2}=\dots=x_{L-3}=b$. Let $x_{L-2}=c$, $x_{L-1}=a$ and $x_L=b$. We observe that for this $(x_1,\dots,x_L)\in [q]^L$ and $\ol\omega$ obtained from $\omega'$ by averaging-out the last two coordinates, \Cref{eqn:maxq_convex} strictly holds. Indeed, since $\omega(L-2)<\omega(L)$ we have 
        \[
            \mathrm{max}_{\omega'}(x_1,\ldots,x_{L-1},x_{L})+\mathrm{max}_{\omega'}(x_1,\ldots,x_L,x_{L-1})=2\sum_{j=1}^{(L-2)/2}\omega(i)+2\omega(L)            
        \]
        and 
        \[
            2\mathrm{max}_{\ol{\omega}}(x_1,\ldots,x_{L-1},x_{L})=2\sum_{j=1}^{(L-3)/2}\omega(i)+\omega(L-1)+\omega(L).
        \]
        Thus \Cref{lem:criterion-for-increase} implies $f(P,\ol\omega) > f(P,\omega')\geq f(P,\omega)$, a contradiction. 

        \item \textbf{$L$ is odd and only one component takes the minimum value.} That is, $\omega(1) = \omega(2) = \cdots = \omega(L-1) < \omega(L)$. Let $\omega'$ be obtained from $\omega'$ by averaging-out the subset $\{L-1,L_2\}$. Then $f(P,\omega') \geq f(P,\omega)$ by \Cref{lem:average-subset-omega} and moreover $\omega'$ is such that at least two coordinates take on the minimum value, as $\omega'(1) = \cdots = \omega'(L-2) > \omega'(L-1)=\omega'(L)$. The argument from the previous case can now be applied to derive a contradiction.
    \end{enumerate}

\end{proof}

Thus, except for degenerate choices for $P \in \Delta([q])$, it follows that the function $\omega \mapsto f(P,\omega)$ is maximized by the choice of $\omega = U_L$. The next step is to determine the distribution $P \in \Delta([q])$ maximizing $P \mapsto f(P,U_L)$. At this point, we can rely on a main result of~\cite{RYZ22}: upon observing that the function $P \mapsto 1-f(P,U_L)$ is the same as the function $f_{q,L}(P)$ defined in~\cite[Equation~(17)]{RYZ22}. It is shown therein that $f_{q,L}(P)$ is \emph{strictly Schur convex}, which in particular means that $f_{q,L}(P)$ has a unique minimum at $P=U_q$. That is, $f(P,U_L)$ has a unique maximum at $P=U_q$.

The (strict) Schur convexity also implies the following: if $p=\max_{x \in [q]}P(x)$, then $f_{q,L}(P) \geq f_{q,L}(P_{q,p})$ where
\begin{align} \label{eq:P_qp-def}
    P_{q,p}(x) = \begin{cases}
        \frac{1-p}{q-1} & x \in \{1,2,\dots,q-1\} \\
        p & x=q 
    \end{cases} \ .
\end{align}
That is, we can conclude that $f(P,U_L) \leq f(P_{q,p},U_L)$. 
We encapsulate these facts in the following proposition. 
 

\begin{proposition}[Theorem 1,2 \cite{RYZ22}]\label{prop:schur_convex}
Let $q\geq 2$, $L \geq q$ and $P \in \Delta([q])$. Suppose $p=\max_{x \in [q]}P(x)$. Then $f(P,U_L)\leq f(P_{q,p},U_L)$. Furthermore, $f(P_{q,p},U_L)\leq f(U_q, U_L)$ is monotone decreasing for $p\geq 1/q$. 
Lastly, $f(P_{q,p},U_L)$ is concave for $p\in [1/q,1]$, i.e., 
$\frac{1}{n}\sum_{i=1}^n f(P_{q,p_i},U_L)\leq f(P_{q,p}, U_L)$ with $p=\frac{1}{n}\sum_{i=1}^n p_i$. 

\end{proposition}
A further fact that we have from~\cite{RYZ22} is that 
\[
    p_*(q,L) = f(U_q,U_L) \ .
\]
In fact, this was taken as the definition of $p_*(q,L)$. 
To end this subsection, we prove the following theorem by utilizing the concavity of $f(P_{q,p},U_L)$.

\begin{theorem}\label{thm:concave}
Assume $\radh(\cC)\leq p$, then we have
$$
\exptover{(\vc_1,\ldots,\vc_L)\in \cC^L}{\rad_\omega(\phi(\vc_1),\ldots,\phi(\vc_L))}\leq f(P_{q,p},U_L).
$$
\end{theorem}
\begin{proof}
Let $\vy$ be the center attaining $\radh(\cC)$. Without loss of generality, we can assume $\vy$ is a all zero vector. 
Let $P_i$ be the distribution of symbols in the $i$-th index of $\cC$, i.e., $P_i(j)=\Pr[\vc(i)=j]$ with the distribution taken over $\vc\in \cC$. Let $p_i=\max_{x\in [q]}P_i(x)$ and $p'=\frac{1}{n}\sum_{i=1}^n p_i$. Clearly, $p_i\geq 1/q$.
Then, we have 
\begin{eqnarray*}
&&\exptover{(\vc_1,\ldots,\vc_L)\in \cC^L}{\rad_{\omega}(\phi(\vc_1),\ldots,\phi(\vc_L))}=
\frac{1}{n}\sum_{i=1}^{n}f(P_i,\omega)\\
&&\leq\frac{1}{n}\sum_{i=1}^{n}f(P_i,U_L)\leq \frac{1}{n}\sum_{i=1}^{n}f(P_{q,p_i},U_L)\leq f(P_{q,p'}, U_L)\leq  f(P_{q,p}, U_L).
\end{eqnarray*}
The first inequality is due to \Cref{lm:convexity} and the second and third inequalities are due to \Cref{prop:schur_convex}. The last inequality is due to $\radh(\cC)\leq p$ and the center $\vy$ is all zero vector. 
The proof is completed.
\end{proof}

%% file: list-dec-abundance.tex

Recall the notations 
$$\rad(\cC)=\frac{1}{n}\min_{\vy\in [q]^n} \max_{\bc\in \cC} \disth(\bc,\vy)$$ 
and
$$\typ_\vu(\vc_1,\ldots,\vc_L) = \frac{1}{n}\sum_{i=1}^n \mathbb{1}\{(\vc_{1}(i),\ldots,\vc_{L}(i))=\bu\}$$ 
where $\bc_i=(\vc_{i}(1),\ldots,\vc_{i}(n))\in [q]^n$ and $\bu\in [q]^L$.
In this subsection, we prove a code $\cC\subseteq [q]^n$ either contains a large subcode $\cC'\subseteq \cC$ with radius $\rad(\cC') \leq 1-\frac{1}{q}-\epsilon$, or most of its $L$-tuples are of type close to the uniform distribution (for all $\vu \in [q]^L)$. 

We first show that for any projection $\pi_A$ with $|A|\geq \mu n$ (for some parameter $\mu \in [0,1]$), the projection $\pi_A(\cC)$ almost preserves the radius $\rad(\cC)$ with small loss. Then, 
if $\rad(\cC')>1-\frac{1}{q}-\epsilon$ for any subcode $\cC'$ with large size, we find a codeword $\vc_1$ in $\cC$ whose symbols' distribution is close to the uniform. In fact, most codewords in $\cC$ satisfies this requirement. Let $A_i$ be the index set of $\vc_1$ taking value $i$. We apply $\pi_{A_i}$ to $\cC$ and claim that $\pi_{A_i}(\cC)$ preserves the radius. Thus, we can find a codeword $\vc_2$ such that for every $i\in [q]$, the symbol's distribution of $\pi_{A_i}(\vc_2)$ is close to uniform. Moreover, most of codewords in $\cC$ satisfy this requirement. The proof is the completed by induction.

\begin{lemma}\label{lm:key}
Let $\pi_A: [q]^n \rightarrow [q]^A$ be the projection on a set $A$ of size $m$. Suppose $\cC\subseteq [q]^n$ is a code of size $qs$ satisfying $\radh(\pi_A(\cC))\leq 1-\frac{1}{q}-\epsilon$. Then, there exists a subcode $\cC'\subseteq \cC$ of size at least $s$ such that $\radh(\cC')\leq 1-\frac{1}{q}-\frac{m}{n}\epsilon$.
\end{lemma}
\begin{proof}
Let $\pi_{\bar{A}}$ be the projection on the remaining $n-m$ coordinates. By the pigeonhole principle, there exists a subcode $\cC'$ of size at least $\frac{|\cC|}{q}$ such that for any codeword $\bc'\in \cC'$, the most frequent symbol of $\pi_{\bar{A}}(\bc')$ is the same. Without loss of generality, we assume this majority symbol is $0$. Let $\by\in [q]^A$ be the center attaining $\radh(\pi_A(\cC))$. 
Define $\vz \in [q]^n$ to be $\vy$ on $A$ and $0$ elsewhere, i.e. 
\[
    z_i = \begin{cases}
        y_i & i \in A\\
        0 & i \notin A
    \end{cases} .
\]
For any codeword $\bc'\in \cC'$, we have
\begin{align*}
\disth(\bc',\bz) &= \disth(\pi_A(\bc'),\by)+\disth(\pi_{\bar{A}}(\bc'),\by')\leq m\left(1-\frac{1}{q}-\epsilon\right)+(n-m)\left(1-\frac{1}{q}\right)\\
&\leq n\left(1-\frac{1}{q}\right)-m\epsilon.
\end{align*}
Thus, $\radh(\cC') \leq 1-\frac{1}{q}-\frac{m}{n}\eps$, as claimed.
\end{proof}
\begin{theorem}\label{thm:radius}
Let $L$ be fixed. For every $\epsilon>0$, there exists a $\delta>0$ with the following property. If $s$ is a natural number, 
there exist constants $M_0=M_0(s)$ and $c(s)$ such that for any code $\cC\subseteq [q]^n$ with size $M\geq M_0$, at least one of the following must hold:
\begin{enumerate}
\item There exists $\cC'\subseteq \cC$ such that $|\cC'|\geq s$ and $\radh(\cC')\leq 1-\frac{1}{q}-\delta$.
\item There exist at least $M^L-c(s)M^{L-1}$ many $L$-tuples of distinct codewords $(\vc_1,\ldots,\vc_L)$ in $\cC$ such that for all $\bu\in [q]^L$ we have
$$
|\type_\vu(\bc_1,\ldots,\vc_L)-q^{-L}|\leq \epsilon
$$
and thus
$$
|\ol{\rad}_\omega(\phi(\vc_1), \cdots, \phi(\vc_L))-f(U_q,\omega)|\leq q^L\epsilon.
$$
\end{enumerate}  
\end{theorem}
\begin{proof}
 Let $\epsilon$ satisfy $\abs{ \paren{ \frac{1}{q}-(q-1)\delta_0 }^L-q^{-L} }\leq \epsilon$ and $\mu=\left(\frac{1}{q}-(q-1)\delta_0\right)^L$, $\delta=\mu\delta_0$. Set $M_0(s)=q^{L+1}s$.
We assume that the first statement does not hold and our goal is to show that the second statement must hold. Since the first statement does not hold, for any $\vy\in [q]^n$, there exists a codeword $\vc\in \cC$ with $\disth(\vc,\vy)>n\paren{ 1-\frac{\ell}{q}-\delta }$. For each $\vc\in \cC$, 
let $\lambda_\vc\in \cX$ be the most frequent symbol of $\vc$.  By the pigeonhole principle, we can find a subcode $\cC'\subseteq \cC$ of size at least $\frac{M}{q}$ such that $\lambda_\vc$ for $\vc\in \cC'$ are the same $\lambda$. Let $\lambda \cdot \mathbf{1}=(\lambda,\lambda,\ldots,\lambda)\in [q]^n$.
It is clear that $\rad(\cC')\leq \frac{1}{n}\max_{\vc\in \cC'} \disth(\vc, \lambda \cdot \mathbf{1}).$ As $M>qs$, this implies $\disth(\vc_1, \lambda \cdot \mathbf{1})>\paren{ 1-\frac{1}{q}-\delta }n$ for some $\vc_1\in \cC'$. (In fact, there exist at least $M-qs$ such $\vc_1$ as we can remove $\vc_1$ from $\cC'$ and obtain the same conclusion.)
Note that necessarily $\disth(\vc_1, \lambda \cdot \mathbf{1})\leq \left(1-\frac{1}{q}\right)n$ (otherwise $\lambda$ would not be the element agreeing the most with $\vc_1$). Let $A_x=\{i\in [n]: c_1(i)=x\}$ for $x\in [q]$. This implies 
$$\frac{|A_x|}{n}\in \brack{ \frac{1}{q}-(q-1)\delta, \frac{1}{q}+\delta }\subseteq \brack{ \frac{1}{q}-(q-1)\delta_0, \frac{1}{q}+\delta_0 },$$
as $\max_{x\in [q]} \frac{|A_x|}{n}\in \left[\frac{1}{q},\frac{1}{q}+\delta\right]$ and $\min_{x\in [q]} \frac{|A_x|}{n}\in \left[\frac{1}{q}-(q-1)\delta,\frac{1}{q}\right]$.  

Now we fix $\vc_1$ and its index set $A_1,\ldots,A_q$ and let $\cC'=\cC\setminus\{\vc_1\}$. We consider the punctured code $\pi_{A_1}(\cC)$. According to \Cref{lm:key}, there exists a subcode $\cC''\subseteq \cC$ of size at most $qs-1$ with $\rad_{\ell}(\pi_{A_1}(\cC''))\leq 1-\frac{\ell}{q}-\delta_0$. Therefore, the same argument as above shows that there exists at least $M-2qs$ codewords $\vc_2\in \cC$ such that the symbol distribution of $\pi_{A_1}(\vc_2)$ is close to uniform, i.e., $|\{i\in A_1: \vc_2(i)=x\}|/|A_1|\in \brack{ \frac{1}{q}-(q-\ell)\delta_0, \frac{1}{q}+\delta_0 }$ for each $x \in [q]$. 
Then, we apply this argument with sets $A_2,\ldots,A_q$ sequentially and conclude that there exists at least $M-2q^2s-1$ codewords $\vc_2\in \cC$ (excluding $\vc_1$) such that the symbol distribution of each $\pi_{A_x}(\vc_2)$ is close to uniform. 

We next partition $[n]$ into $q^2$ sets $A_{xy}=\{i\in [n]: \vc_1(i)=x,\vc_2(j)=y\}$ for $x,y \in [q]$ according to the value of $\vc_1$ and $\vc_2$. This gives $\frac{|A_{xy}|}{n}\in \brack{ \paren{ \frac{1}{q}-(q-1)\delta_0 }^2, \paren{ \frac{1}{q}+\delta_0 }^2 }$.
One can continue this process and construct $L$-tuples $\vc_1,\ldots,\vc_L$ for which necessarily
\begin{align} \label{eq:random-like}
    \forall \vu \in [q]^L,~~~|\type_\vu(\bc_1,\ldots,\vc_L)-q^{-L}|\in \brack{ \left(\frac{1}{q}-(q-1)\delta_0\right)^L, \paren{ \frac{1}{q}+\delta_0 }^L }
\end{align}
In general, there are more than 
$$
N_1=\prod_{i=0}^{L-1}(M-j-2q^{i+1}s)
$$
$L$-tuples $(\bc_1,\ldots,\vc_L)$ satisfying \Cref{eq:random-like}. This implies $N_1\geq M^L-cM^{L-1}$ where $c$ only depends on $q$ and $L$. The proof is completed.
\end{proof}

%% file: list-dec-end.tex
The argument follows the same line of reasoning as~\cite{abp-2018}. We provide the proof for completeness. Define $\rho_L(\cC)=\min \rad(\phi(\vc_1),\ldots,\phi(\vc_L))$ with minimum taken over all $L$-tuples $(\vc_1,\ldots,\vc_L)\in \cC^L$ with distinct elements, where we recall that $\rad$ denotes the \emph{relaxed} Chebyshev radius (\Cref{def:def-relaxed-cheb-rad}). 

\begin{theorem}
Let $L\geq 2$ and $q\geq 3$. If $\cC\subseteq [q]^n$ is $(p_*(q,L)+\epsilon, L)$-list-decodable, then $|\cC|=  O_{q,L}(\frac{1}{\epsilon})$. 
\end{theorem}


\begin{proof}
Shorthand $\tau_L=p_*(q,L)$ and $\tau_{p,L}=f(P_{q,p},U_L)$ with $P_{q,p}$ defined in \Cref{prop:schur_convex}. 
Note that $\tau_{p,L}<\tau_L$ if $p>1/q$.

Our first step is to obtain a subcode $\cC_1\subseteq \cC$ with $\rho(\cC_1)\geq \tau_L+\epsilon$. By the list-decodability assumption on $\cC$, for any $L$-tuple $(\vc_1,\ldots,\vc_L)\in \cC^L$ with distinct elements, we have 
$\radh(\vc_1,\ldots,\vc_L)\geq \tau_L+\epsilon$. To apply \Cref{lm:convert}, we want this Hamming metric radius slightly larger. To do this, we find a subcode $\cC_1\subseteq \cC$ such that all codewords in $\cC_1$ have the same prefix of length $rL$ where $r=\lfloor 1/(\tau_L+\eps)\rfloor$. By the pigeonhole principle, $|\cC_1|\geq q^{-rL}|\cC|$. Removing these $rL$ indices we obtain a code $\cC_2$ for which, for all $(\vc_1,\dots,\vc_L) \in \cC_2^L$ with distinct elements, we have
\[
    \radh(\vc_1,\dots,\vc_L) \geq \frac{n}{n-rL}(\tau_L+\eps) \geq \left(1+\frac{rL}{n}\right)\tau_L+\eps \geq \tau_L + \eps + \frac{L}{n} \ .
\]
Applying \Cref{lm:convert}, we find that $\rho_L(\cC_2)\geq \tau_L+\epsilon$ or equivalently $\rad(\vc_1,\ldots,\vc_L)\geq \tau_L+\epsilon$ for any $L$-tuple $(\vc_1,\ldots,\vc_L)$ with distinct elements. 
We divide our discussion into two cases.
\begin{itemize}
\item \label{ramsey-listdec} Suppose $\radh(\cC_2)\leq 1-\frac{1}{q}-\delta$ for some constant $\delta>0$. Let $p=\frac{1}{q}+\delta$ and $\vy$ the center attaining $\radh(\cC_2)$. 
By ordering the elements of $\cC$ arbitrarily, we may identify ordered $L$-element tuples of $\cC^L$ with distinct elements with $L$-element subsets of $\cC$. For every such ordered $L$-tuple $(\vc_1,\ldots,\vc_L)$, there is a weight $\omega\in \Omega_L$ that solves 
$$\rad(\phi(\vc_1),\ldots,\phi(\vc_L))=\max_{\omega\in \Omega_L}\ol{\rad}_\omega(\phi(\vc_1), \cdots, \phi(\vc_L)).$$
Each solution $\omega$ gives a coloring of $L$-element subsets of $\cC_2$: we assign color $\omega$ to the $L$-element subset if $\omega$ is a maximizer (breaking ties arbitrarily). 
As this coloring has at most $|\Omega_L|$ colors and \Cref{lm:finite} promises $|\Omega_L| = O_{q,L}(1)$ (in particular, it's finite), by the hypergraph version of Ramsey's theorem~\cite[Theorem~2]{graham1991ramsey} it follows that if $\cC_2$ is large enough, there is a monochromatic subset $\cC_3\subseteq \cC_2$ of size exceeding $L^2/(\tau_L-\tau_{p,L})$. 

On the other hand, let $\cT$ be the set of all ordered $L$-tuples of distinct codewords in $\cC_3$. If $(\vc_1,\ldots,\vc_L)$ is an $L$-tuple selected uniformly at random in $\cC_3^{L}$, then $\Pr[(\vc_1,\ldots,\vc_L)\notin \cT]\leq \frac{\binom{L}{2}}{|\cC_3|}<\tau_L-\tau_{p,L}$. 
Since $\radh(\cC_3)\leq 1-p$, by \Cref{thm:concave} and \Cref{prop:schur_convex}, we have 
\begin{align*}
    \tau_{p,L}&\geq \exptover{(\vc_1,\ldots,\vc_L)\in \cC_3^L}{\ol{\rad}_\omega(\phi(\vc_1), \cdots, \phi(\vc_L))} \\
    &\geq \Pr[(\vc_1,\ldots,\vc_L)\in \cT]\exptover{(\vc_1,\ldots,\vc_L)\in \cT}{\ol{\rad}_\omega(\phi(\vc_1), \cdots, \phi(\vc_L))}.
\end{align*}
This implies there exists an $L$-tuple of distinct codewords $\vc_1,\ldots,\vc_L$ in $\cC_3$ such that $$(1-\tau_L+\tau_{p,L})\ol{\rad}_\omega(\phi(\vc_1), \cdots, \phi(\vc_L))< \tau_{p,L}.$$ 
It follows that 
$$
\ol{\rad}_\omega(\phi(\vc_1), \cdots, \phi(\vc_L))<\tau_{p,L}+\tau_L-\tau_{p,L}=\tau_L
$$
as $\ol{\rad}_\omega(\phi(\vc_1), \cdots, \phi(\vc_L))\leq 1$.
Contradiction.

\item Otherwise, let $\cH$ be the collection of all $L$-tuples $(\vc_1,\ldots,\vc_L)$ in $\cC_2^L$ such that 
$$\ol{\rad}_\omega(\phi(\vc_1), \cdots, \phi(\vc_L))>\tau_L$$ 
for some $\omega\neq U_L$.
Let $\epsilon_0=q^{-L}\min\{\tau_L-f(U_q,\omega): \omega\in \Omega_L\}$; since $\omega \mapsto f(U_q,\omega)$ is uniquely maximized by $U_L$ (\Cref{thm:maxq-unique}) and $\Omega_L$ is finite (\Cref{lm:finite}), $\eps_0>0$. 
By \Cref{thm:radius} applied with $\eps=\eps_0$, there exist at least $|\cC_2|^L-c|\cC_2|^{L-1}$ many $L$-tuples of distinct codewords $\vc_1,\ldots,\vc_L$ in $C$ such that 
$$\ol{\rad}_\omega(\phi(\vc_1), \cdots, \phi(\vc_L))\leq f(U_q,\omega)+q^L\epsilon_0\leq\tau_L.$$
Thus, $|\cH|\leq c |\cC_2|^{L-1}$ where $c$ depends only on $q$ and $L$. Let $(\bc_1,\ldots,\bc_L)$ be a random $L$-tuple in $\cC_2^L$. The probability that $(\bc_1,\ldots,\bc_L)$ are $L$ distinct codewords is at least $1-\frac{\binom{L}{2}}{|\cC_2|}=O(\frac{1}{|\cC_2|})$. The probability that $(\bc_1,\ldots,\bc_L)\in \cH$ is at most $O(\frac{1}{|\cC_2|})$. This implies 
$\Pr[(\vc_1,\ldots,\vc_L)\in \cT\setminus \cH]\geq 1-O(\frac{1}{|\cC_2|})$. Thus, 
\begin{eqnarray*}
\tau_L&\geq&\exptover{(\vc_1,\ldots,\vc_L)\in \cC_2^L}{\ol{\rad}_{U_L}(\phi(\vc_1), \cdots, \phi(\vc_L))}\\
 &\geq&
  \Pr[(\vc_1,\ldots,\vc_L)\in \cT\setminus\cH]\exptover{(\vc_1,\ldots,\vc_L)\in \cT\setminus \cH}{\ol{\rad}_{U_L}(\phi(\vc_1), \cdots, \phi(\vc_L))}\\
  &\geq&\left(1-O\left(\frac{1}{|\cC_2|}\right)\right)\exptover{(\vc_1,\ldots,\vc_L)\in \cT\setminus \cH}{\ol{\rad}_{U_L}(\phi(\vc_1), \cdots, \phi(\vc_L))} \ .
 \end{eqnarray*}
On the other hand, for any $(\vc_1,\dots,\vc_L) \in \cT \setminus \cH$, 
$$\ol{\rad}_{U_L}(\phi(\vc_1), \cdots, \phi(\vc_L))\geq \rho(\cC_2)\geq \tau_L+\epsilon.$$
This implies that $|\cC_2|\leq O_{q,L}(\frac{1}{\epsilon})$ and thus also $|\cC|\leq  O_{q,L}(\frac{1}{\epsilon})$. \qedhere
\end{itemize}
\end{proof}

%% file: list-rec-radius.tex
First, we can similarly prove that $\rad$ is close to $\rad_\ell$ by designing a linear programming relaxation.
\begin{lemma}[$\rad$ is close to $ \rad_\ell $]\label{lm:lrconvert}
Let $ \vc_1, \cdots, \vc_L \in [q]^n $. 
Denote by $ \vx_1, \cdots, \vx_L\in(\partial\Delta_\ell)^n $ the images of $ \vc_1, \cdots, \vc_L $ under the embedding $\phi_\ell$. 
Then 
\begin{align}
\rad_\ell(\vc_1, \cdots, \vc_L)
&\le \rad(\vx_1, \cdots, \vx_L) +\frac{L}{n}. \notag 
\end{align}
\end{lemma}

\begin{proof}
Suppose $ n\cdot\rad(\vx_1, \cdots, \vx_L) = t $. 
Then there exists $ \vy\in\Delta_\ell^n $ such that for every $ i\in[L] $, 
$$
d(\vx_i, \vy)
\le t , \notag 
$$
That is, the following polytope is nonempty:
\begin{multline}
\brace{
    \vy \in \Delta^n : 
    \forall i\in[L], \; d(\vx_i, \vy) \le t
} \\
= \brace{
    (y(j,A))_{(j, A)\in[n]\times \cX} : \begin{array}{l}
    \forall (j,k)\in[n]\times \cX,\, y(j,A) \ge 0 , \\
    \forall j\in[n], \, \sum_{A\in \cX} y(j,A) = 1 , \\ 
    \forall i\in[L], \, \frac{1}{2} \sum_{j = 1}^n \paren{ \sum_{A\in \cX_{\vc_i(j)}}(1 - y(j, A)) + \sum_{A\in \cX\setminus\cX_{\vc_i(j)}} y(j,A)-\binom{q-1}{\ell-1}+1} \le t
    \end{array}
} . \notag 
\end{multline}

Equivalently, the following linear program (LP) is feasible: 
\begin{align}
\begin{array}{cl}
\max\limits_{(y(j,A))_{(j, A)\in[n]\times \cX}} & 0 \\
\mathrm{s.t.} & \forall (j, A)\in[n]\times \cX, \, y(j,A) \ge 0 , \\
    & \forall j\in[n], \, \sum_{A\in \cX} y(j,A) = 1 , \\ 
    & \forall i\in[L], \, \frac{1}{2} \sum_{j = 1}^n \paren{ \sum_{A\in \cX_{\vc_i(j)}}(1 - y(j, A)) + \sum_{A\in \cX\setminus\cX_{\vc_j(j)}} y(j,A) -\binom{q-1}{\ell-1}+1} \le t . 
\end{array}
\notag 
\end{align}
Similar to the list decoding case, our LP  can be written as 
\begin{align}
\begin{array}{cl}
\max\limits_{\vy\in\bbR^{\binom{q}{\ell}n}, \vz\in\bbR^L} & 0 \\
\mathrm{s.t.} &  \matrix{A & I_L \\ B & 0} \matrix{\vy \\ \vz} = \matrix{t \one_L \\ \one_n} , \\
    & \vy,\vz \ge \zero . 
\end{array}
\notag 
\end{align}
where $A$ is an $L\times \binom{q}{\ell}n$ matrix and $B$ is an $n\times \binom{q}{\ell}n$ matrix encoding the third and second constraints in the LP problem. By \Cref{prop:basic_feasible_sol}, there exists a basic feasible solution $\vy,\vz$ with at most $n+L$ nonzero components and thus $\vy=(y(j,A))_{(j, A)\in[n]\times \cX}$ has at most $n+L$ nonzeros. Since $\sum_{A\in \cX} y(j,T) = 1$, at least one of $y(j,T), T\in \cX$ are nonzero. This implies that there are at least $n-L$ choices for $j\in [n]$ such that $(y(j,T))_{T\in \cX}=\ve_S$ for some $S\in \cX$. We proceed to construct the set $\vY=(Y_1,\ldots,Y_n)\in \cX^n$. If $\by_j=\ve_S$, we set $Y_j=S$. Since there are at least $n-L$ indices $j$ with $\by_j=\ve_S$ for some $S\in \cX$,
we have at most $L$ $Y_j$ yet to be determined which are set to be any $\ell$-subsets in $\cX$. By construction, the difference between $\distlr(\bc_i,\vY)$ and $d(\bx_i,\by)$ is at most $L$. The proof is completed.
\end{proof}

We further relax rad by defining the weighted average $\ell$-radius. For $\vx_1,\ldots,\vx_L\in (\partial \Delta_\ell)^n$ and $\omega\in \Delta([L])$, let 

\begin{align}
\ol{\rad}_{\omega,\ell}(\vx_1, \cdots, \vx_L) &\coloneqq \frac{1}{n}\min_{\vy\in\Delta_\ell^n} \exptover{i\sim \omega}{d(\vx_i, \vy)}
= \frac{1}{n}\min_{\vy\in\Delta_\ell^n} \sum_{i\in[L]} \omega(i) d(\vx_i, \vy) . \notag 
\end{align}

We can minimize each component of $\vy(A)$ so as to obtain the minimization of the above function. Suppose $\bx_1,\ldots,\bx_L\in (\partial \Delta_\ell)^n$ are the images of $\vc_1,\ldots,\vc_L\in [q]^n$ under the embedding $\phi_\ell$. Then,

\begin{align}
&\ol{\rad}_{\omega,\ell}(\vx_1, \cdots, \vx_L)
=\frac{1}{n} \min_{\vy\in\Delta_\ell^n} \sum_{i\in[L]} \omega(i) d(\vx_i, \vy) \notag \\
&= \frac{1}{2n} \min_{\vy\in\Delta_\ell^n} \sum_{i\in[L]} \omega(i) \sum_{j = 1}^n \paren{ \sum_{A\in \cX_{\vc_i(j)}}(1 - \vy_j(A)) + \sum_{A'\in \cX \setminus\cX_{\vc_i(j)}} \vy_j(A')-\binom{q-1}{\ell-1}+1} \notag \\
&= \frac{1}{2n} \min_{(\vy_1,\ldots,\vy_n)\in\Delta_\ell^n} \sum_{j = 1}^n \brack{ \sum_{i\in[L]} \omega(i)\bigg(\sum_{A\in \cX_{\vc_i(j)}}(1 - \vy_j(A)) + \sum_{A'\in \cX \setminus\cX_{\vc_i(j)}} \vy_j(A')-\binom{q-1}{\ell-1}+1\bigg)} \notag \\
&= \frac{1}{2n} \sum_{j = 1}^n \min_{\vy_j\in\Delta_\ell} \brack{ \sum_{i\in[L]} \omega(i)(2 - 2\sum_{A\in \cX_{\vc_i(j)}} \vy_j(A))} \notag \\
&= \frac{1}{n}\sum_{j = 1}^n \min_{\vy_j\in\Delta_\ell} \brack{ 1 - \sum_{i\in[L]} \omega(i) \sum_{A\in \cX_{\vc_i(j)}} \vy_j(A) } \notag \\
&= 1 - \frac{1}{n}\sum_{j = 1}^n \max_{\vy_j\in\Delta_\ell} \brack{ \sum_{i\in[L]} \omega(i) \sum_{A\in \cX_{\vc_i(j)}} \vy_j(A)} \notag \\
&= 1 - \frac{1}{n}\sum_{j = 1}^n \max_{A\in \cX} \sum_{\substack{i\in[L] \\ \vc_i(j)\in A}} \omega(i) . \label{eqn:lr_avg_rad_1} 
\end{align}
\Cref{eqn:lr_avg_rad_1} is due to the fact that the maximizer $\vy^*$ to the following problem 
$$
\max_{\vy\in\Delta_\ell} \sum_{i\in[L]} \omega(i) \sum_{A\in \cX_{x_i}} \vy(A)
=\max_{\vy\in\Delta_\ell}\sum_{A\in \cX} \vy(A)
\sum_{\substack{i\in [L]\\ x_i\in A} }\omega(i).
$$
is obtained from a set $A\in \cX$ that maximizes 
$\sum_{\substack{i\in [L], x_i\in A}} \omega(i)$
and setting $\vy^*=\ve_A$.

Clearly 
$$
\ol{\rad}_{\omega,\ell}(\vx_1, \cdots, \vx_L)\leq \rad(\vx_1,\ldots,\vx_L). 
$$
Similarly, we can prove that rad is equal to the maximum $\ol{\rad}_{\omega,\ell}$ over $\omega$. 
\begin{lemma}[$ \rad $ equals maximum $ \ol{\rad}_\omega $]\label{lm:lrconvexity}
For any $\vx_1,\ldots,\vx_L\in (\partial \Delta_\ell)^n$, 
$$
\rad(\vx_1,\ldots,\vx_L)=\max_{\omega\in \Delta([L])}\ol{\rad}_{\omega,\ell}(\vx_1,\ldots,\vx_L).
$$
\end{lemma}
\begin{proof}
The proof is exactly the same as in \Cref{lm:convexity}. We observe that 
\begin{align}
\rad(\vx_1, \cdots, \vx_L) &\coloneqq \frac{1}{n}\min_{\vy\in\Delta_\ell^n} \max_{i\in[L]} d(\vx_i, \vy) 
= \frac{1}{n}\min_{\vy\in\Delta_\ell^n} \max_{\omega\in\Delta([L])} \sum_{i\in[L]} \omega(i) d(\vx_i, \vy) , \notag 
\end{align}
Then, we apply von Neumann's minimax theorem to interchange min and max. 
\end{proof}
\begin{lemma}[$\rad$ is achieved by finitely many $\omega$]\label{lm:lrfinite}
For every $L$ and $\ell$, there exists a finite set of probability measure $\Omega_{\ell,L}\subseteq \Delta([L])$ such that 
$$
\rad(\vx_1,\ldots,\vx_L)=\max_{\omega\in \Omega_L}\ol{\rad}_{\omega,\ell}(\vx_1,\ldots,\vx_L).
$$
for all $\vx_1,\ldots,\vx_L\in \partial \Delta_\ell^n$. 

 \end{lemma} 
 
\begin{proof}
 The proof is very similar to \Cref{lm:finite} except that our signatures are now maps $[q]^L \to \binom{[q]}{\ell}$ which yields $\binom{q}{\ell}q^L$ hyperplanes $H_{A,\bu} = \{\omega \in \Delta([L]) \colon \sum_{i:\bu(i) \in A}\omega(i) \geq 1/\binom{q}{\ell}\}$ for each $A \in \binom{[q]}{\ell}$ and $\bu \in [q]^L$. Similarly, to see that the sets $\Omega_S = \{\omega \in \Delta([L])\colon S \text { is a signature for } \omega\}$ for $S:[q]^L \to \binom{[q]}{\ell}$ are indeed convex polytopes one must simply note that it is an intersection over $\vu \in [q]^L$ of the convex polytopes 
 \begin{align*}
    \left\{\omega \in \Delta([L])\colon \exists \text{ signature } \right.&\left.S_\omega \text{ for } \omega \text{ s.t. }S_\omega(\vu) = S(\vu)\right\} \\
    &= \bigcap_{B \in \binom{[q]}{\ell}\setminus S(\vu)} \left\{\omega \in \Delta([L])\colon \sum_{i:\vu(i)\in S(\vu)}\omega(i) \geq \sum_{i:\vu(i)\in B}\omega(i)\right\} \ .
\end{align*}
Lastly, the argument now uses \Cref{eqn:lr_avg_rad_1} in order to view the optimization as a linear program.
\end{proof}

%% file: list-rec-function-property.tex
We consider the expected weighted average $\ell$-radius of a sequence of i.i.d.\ symbols. Let 
\begin{align}
f_\ell(P, \omega) &\coloneqq \exptover{(X_1, \cdots, X_L)\sim P^{\ot L}}{\ol{\rad}_{\omega,\ell}(\ve_{X_1}, \cdots, \ve_{X_L})}. \notag 
\end{align}
where $P$ is a probability distribution over $[q]$. Using \Cref{eqn:lr_avg_rad_1}, we have
\begin{align}
f_\ell(P, \omega)&= \exptover{(X_1, \cdots, X_L)\sim P^{\ot L}}{1-\max_{A\in \cX}\sum_{\substack{i\in [L]\\ X_i\in A}} \omega(i)} \notag \\
&=1-\sum_{(x_1,\ldots,x_L)\in [q]^L}\bigg(\prod_{i=1}^L P(x_i)\bigg)\max_{A\in \cX}\sum_{i\in [L]}\omega(i)\indicator{x_i\in A}.
\end{align}

We define the shorthand notation
$$
\mathrm{max}_{\omega,\ell}(x_1,\ldots,x_L):=\max_{A\in \cX}\sum_{ \substack{i\in [L] \\ x_i\in A}} \omega(i), \qquad \mathrm{argmax}_{\omega,\ell}(x_1,\ldots,x_L):=\argmax_{A\in \cX}\sum_{ \substack{i\in [L] \\ x_i\in A}} \omega(i).
$$
for any $\omega\in \Delta([L])$ and $(x_1,\ldots,x_L)\in [q]^L$.

We again begin by establishing that $f_\ell(P,\omega)$ cannot decrease if $\omega$ is replaced by $U_L$, in analogy to \Cref{lm:convexity}.

First, we give a criterion for increase, which is analogous to  \Cref{lem:criterion-for-increase}. The proof is completely analogous to the previous proof and is therefore omitted. 

\begin{lemma} \label{lem:criterion-for-increase-lr}
    Let $P \in \Delta([q])$ and $\omega \in \Delta([L])$. Suppose $\omega(L-1)\neq \omega(L)$ and that $\ol\omega \in \Delta([L])$ is obtained by averaging-out the last two coordinates of $\omega$. Suppose that for all $(x_1,\dots,x_L) \in [q]^L$ we have 
    \begin{align}
        \frac{1}{2} \paren{ \mathrm{max}_{\omega,\ell}(x_1, \cdots, x_{L-1}, x_L) + \mathrm{max}_{\omega,\ell}(x_1, \cdots, x_{L}, x_{L-1}) }
        &\ge \mathrm{max}_{\ol{\omega},\ell}(x_1, \cdots,x_{L-1}, x_L) . \label{eqn:maxlr} 
    \end{align}
    Then $f(P,\ol\omega) \geq f(P,\omega)$. 
    
    Furthermore, suppose that additionally there exists $(x_1,\dots,x_L) \in [q]^L$ with $\prod_{i=1}^LP(x_i)>0$ such that the inequality in \Cref{eqn:maxlr} is strict. Then $f(P,\ol \omega) > f(P,\omega)$. 
\end{lemma}

\begin{lemma}
Let $\ell\geq 2$ and $q\geq 3$ with $\ell \leq q$. For any distribution $P$ and $\omega\in \Delta([L])$, $f_{\ell}(P,\omega)\leq f_{\ell}(P, U_L)$.
\end{lemma}
\begin{proof} 
Fix any $ (x_1, \cdots, x_L)\in[q]^L $. 
Let $ \omega\in\Delta([L]) $ be non-uniform. 
Without loss of generality, assume $ \omega(L-1)\neq\omega(L)$ and let $\ol{\omega}$ be obtained from $\omega$ by averaging-out the last two coordinates. We will show \Cref{eqn:maxlr} holds, which suffices by \Cref{lem:criterion-for-increase-lr}.

\Cref{eqn:maxlr} clearly holds if $x_{L-1}=x_L$. We now assume $x_L=a$ and $x_{L-1}=b$ with $a\neq b$.
Let 
$$
\omega^{(a)}=\sum_{\substack{i\in[L-2] \\ x_i = a}}\omega(i), \qquad \omega^{(b)}=\sum_{\substack{i\in[L-2] \\ x_i = a}}\omega(i) \ .
$$
Let $T=\mathrm{argmax}_{\ol{\omega},\ell}(x_1,\ldots,x_L)$. If $x_{L-1}, x_L\in T$, then \Cref{eqn:maxlr} holds trivially. 
Otherwise, we assume at least one of them is not in $T$. Without loss of generality, we first assume that $a\in T$ and $b\notin T$.
This implies that $\sum_{x_i=a}\ol{\omega}(i)\geq \sum_{x_i=b}\ol{\omega}(i)$ and $\omega^{(a)}\geq \omega^{(b)}$. In this case, we have
$$
2\mathrm{max}_{\ol{\omega},\ell}(x_1,\ldots,x_L)=2\sum_{x_i\in T\setminus\{a\}}\omega(i)+2\omega^{(a)}+\omega(L-1)+\omega(L).
$$
On the other hand, 
$$
\mathrm{max}_{\omega,\ell}(x_1,\ldots,x_{L-1},x_L)\geq 
\sum_{x_i\in T\setminus\{a\}}\omega(i)+\omega^{(a)}+\omega(L),
$$
and
$$
\mathrm{max}_{\omega,\ell}(x_1,\ldots,x_{L},x_{L-1})\geq 
\sum_{x_i\in T\setminus\{a\}}\omega(i)+\omega^{(a)}+\omega(L-1).
$$
Thus \Cref{eqn:maxlr} holds. It remains to consider the case both $a,b$ are not in $T$. In this case, we observe that 
$$
\mathrm{max}_{\omega,\ell}(x_1,\ldots,x_{L-1},x_L)\geq 
\sum_{x_i\in T}\omega(i)=
\sum_{x_i\in T}\ol{\omega}(i)=\mathrm{max}_{\ol{\omega},\ell}(x_1,\ldots,x_L)
$$
and
$$
\mathrm{max}_{\omega,\ell}(x_1,\ldots,x_{L},x_{L-1})\geq 
\sum_{x_i\in T}\omega(i)=
\sum_{x_i\in T}\ol{\omega}(i)=\mathrm{max}_{\ol{\omega},\ell}(x_1,\ldots,x_L).
$$
Thus \Cref{eqn:maxlr} always holds. 
\end{proof}

\begin{theorem}\label{thm:lrmaxq}
Let $q>\ell\geq 2$, $L > \ell$ and $\omega \in \Delta([L])$. Assume $P(x)>0$ for all $x\in [q]$.\footnote{This theorem is only applied with $P=U_q$ which clearly satisfies this condition.} Then $f_\ell(P, \omega) = f_\ell(P, U_L)$ if and only if $ \omega = U_L$. 
\end{theorem}
\begin{proof}
To prove this claim, by \Cref{lem:criterion-for-increase-lr} it suffices to show that for any non-uniform $\omega\in \Delta([L])$ there exists a $(x_1,\ldots,x_L)\in [q]^L$, such that \Cref{eqn:maxlr} strictly holds. 
Since $L>\ell\geq 2$, assume $L=\ell+a\geq 3$ with $a\geq 1$. 

Suppose there exist two components $\omega(L),\omega(L-1)$ achieving the minimum $\min_{i\in [L]} \omega(i)$. Let $\omega(j)$ be the maximum $\max_{i\in [L]} \omega(i)$, which is necessarily strictly larger than $\omega(L-1)$ and $\omega(L)$ (otherwise $\omega=U_L$). Let $\omega'$ be obtained from $\omega$ by averaging-out the coordinates $L-1$ and $j$. \Cref{lm:lrconvexity} promises $f_\ell(P,\omega)\leq f_{\ell}(P,\omega')$. Clearly, $\omega'(L-1)>\ol{\omega}(L)$.
We can continue this process until there is only one component achieving the minimum of $\omega$ (note that we will never obtain the uniform distribution, as it will always hold that $\omega(L)<1/q$). 

We may now assume $\omega(L)$ is smaller than any other component of $\omega$. Let $x_1=x_2=\cdots=x_{a}=1$ and $x_{a+i}=i+1$ for $i=1,\ldots,\ell$. 
This can be done as $q>\ell\geq2$. It is clear that 
$$
\mathrm{max}_{\omega,\ell}(x_1, \cdots,x_{L-1}, x_L)=\sum_{i=1}^{L-1}\omega(i),\qquad  \mathrm{max}_{\omega,\ell}(x_1, \cdots, x_L, x_{L-1})=\sum_{i=1}^{L-1}\omega(i).
$$
Let $ \ol{\omega}\in\Delta([L]) $ be obtained from $\omega$ by averaging-out the last two coordinates.
Then 
$$
2\mathrm{max}_{\ol{\omega},\ell}(x_1, \cdots,x_{L-1}, x_L)=2\sum_{i=1}^{L-2}\omega(i)+\omega(L-1)+\omega(L).
$$
As $\omega(L)<\omega(L-1)$, \Cref{eqn:maxlr} strictly holds. By \Cref{lem:criterion-for-increase-lr}, the proof is completed.
\end{proof}

Now, we turn to maximizing the other derived function $P \mapsto f_{\ell}(P,U_L)$. The function $f_\ell(P,U_L)$ is the same as the function $1-f_{q,L,\ell}(P)$ from~\cite{RYZ22} where
\begin{align*}
    f_{q,L,\ell}(P) &\coloneqq \exptover{(X_1,\cdots,X_L)\sim P^{\ot L}}{\plur_\ell(X_1,\cdots,X_L)} .
\end{align*}
As before, we can rely on certain Schur convexity and convexity results from \cite{RYZ22}. 

\begin{proposition}[Theorem 8, 9~\cite{RYZ22}]\label{prop:lrschur_convex}
For any $q>\ell\geq 2$ and $L > \ell$, $f_{\ell}(P,U_L)\leq f_{\ell}(P_{q,\ell,p},U_L)$ where $p=\max_{A\in \cX}\sum_{i\in A}P(i)$ and 
\begin{align}
    P_{q,\ell,p}(i) &= \begin{cases}
    \frac{1-p}{q-\ell} , & 1\le i\le q-\ell \\
    \frac{p}{\ell} , & q-\ell + 1\le i\le q
    \end{cases} . \label{def:p-q-ell-rho}
\end{align}
 $f_{\ell}(P_{q,\ell,p},U_L)\leq f_{\ell}(U_q,U_L)$ is monotone decreasing for $p\geq \ell/q$. Moreover $f_{\ell}(P_{q,\ell,p}, U_L)$ is concave for $p\in [0,1]$, i.e., $\frac{1}{n}\sum_{i=1}^n f_\ell(P_{q,\ell,p_i},U_L)\leq f_\ell(P_{q,\ell,p}, U_L)$ with $p'=\frac{1}{n}\sum_{i=1}^n p_i$. 
\end{proposition}

As a corollory of the above we can prove the following. 
\begin{theorem}\label{thm:lrconcave}
Assume $\rad_\ell(\cC)\leq p$, then we have 
$$
\exptover{(\vc_1,\ldots,\vc_L)\in \cC^L}{\rad_{\omega,\ell}(\phi_\ell(\vc_1),\ldots,\phi_\ell(\vc_L))}\leq f(P_{q,\ell,p},U_L).
$$
\end{theorem}
\begin{proof}
Let $\vY$ be the center attaining $\rad_{\ell}(\cC)$. Without loss of generality, we assume $\vY=(\{q-\ell+1,\ldots,q\},\ldots,\{q-\ell+1,\ldots,q\})$. 
Let $P_i$ be the distribution of symbols in the $i$-th index of $\cC$, i.e., $P_i(j)=\Pr[\vc(i)=j]$ where $\vc\in \cC$ is a random codeword distributed uniformly at random. Let $p_i=\max_{A\in \cX}\sum_{j\in A}P_i(j)$ and $p'=\frac{1}{n}\sum_{i=1}^n p_i$.
Then, we have 
\begin{eqnarray*}
&&\exptover{(\vc_1,\ldots,\vc_L)\in \cC^L}{\rad_{\omega,\ell}(\phi_\ell(\vc_1),\ldots,\phi_\ell(\vc_L))}=
\frac{1}{n}\sum_{i=1}^{n}f_\ell(P_i,\omega)\\
&&\leq\frac{1}{n}\sum_{i=1}^{n}f_\ell(P_i,U_L)\leq \frac{1}{n}\sum_{i=1}^{n}f_\ell(P_{q,\ell,p_i},U_L)\leq f_\ell(P_{q,\ell,p'}, U_L)\leq f(P_{q,\ell,p}, U_L).
\end{eqnarray*}
The first inequality is due to \Cref{lm:lrconvexity} and the second and third inequalities are due to \Cref{prop:lrschur_convex}. The last inequality is due to $\rad_\ell(\cC)\leq p$ and the form of the center $\vY$.
The proof is completed.
\end{proof}

%% file: list-rec.tex
Recall $\rad_{\ell}(\cC)=\frac{1}{n}\min_{\vY\in \cX^n} \max_{\bc\in \cC} \distlr(\bc,\vY)$. The goal now is to show that, unless $\cC$ has a large biased subcode, most $L$-tuples in $\cC$ have a random-like type. The proof is quite similar to the list-decoding case; we present here for completeness.

\begin{lemma}\label{lm:lrkey}
Let $\pi_A: [q]^n \rightarrow [q]^A$ be the projection on a set $A$ of size $m$. Suppose $\cC\subseteq [q]^n$ is a code of size $\binom{q}{\ell}s$ satisfying $\rad_\ell(\pi_A(\cC))\leq 1-\frac{\ell}{q}-\epsilon$. Then, there exists a subcode $\cC'\subseteq \cC$ of size at least $s$ such that $\rad_\ell(\cC')\leq 1-\frac{\ell}{q}-\frac{m}{n}\epsilon$.
\end{lemma}
\begin{proof}
The proof is similar as in \Cref{lm:key}. Let $\pi_{\bar{A}}$ be the projection on the remaining $n-m$ indices. By the pigeonhole principle, there exist a subcode $\cC'$ of size at least $\frac{|\cC|}{\binom{q}{\ell}}$ such that the $\ell$ most frequent symbols of  $\pi_{\bar{A}}(\vc')$ is the same. Let $T$ be this set and we have $\distlr(\pi_{\bar{A}}(\vc'), \vT)\leq 1-\frac{\ell}{q}$ with $\vT=(T,T,\ldots,T)\in \cX^{\bar{A}}$. Let $\vY\in \cX^A$ be the center attaining $\rad_\ell(
\pi_A(\cC))$. Define $\vZ$ as $\pi_A(\vZ)=\vT,\pi_{\bar{A}}(\vZ)=\vY$ and for any codeword $\vc'\in \cC'$, we have
\begin{align*}
    \distlr(\vc',\vZ) &= \distlr(\pi_A(\vc),\vY)+\distlr(\pi_{\bar{A}}(\vc'),\vT) \\
    &\leq m\left(1-\frac{\ell}{q}-\epsilon\right)+(n-m)\left(1-\frac{\ell}{q}\right)\leq n\left(1-\frac{\ell}{q}\right)-m\epsilon. \qedhere
\end{align*}
\end{proof}

\begin{theorem}\label{thm:lrradius}
Let $q,L,\ell$ be fixed. For every $\epsilon>0$, there exists a $\delta>0$ with the following property. If $s$ is a natural number, 
there exist constants $M_0=M_0(s)$ and $c(s)$ such that for any code $\cC\subseteq [q]^n$ with size $M\geq M_0$, one of the following 
two alternatives must hold:
\begin{enumerate}
\item There exists $\cC'\subseteq \cC$ such that $|\cC'|\geq s$ and $\rad_\ell(\cC')\leq 1-\frac{\ell}{q}-\delta$.
\item There exists at least $M^L-c(s)M^{L-1}$ many $L$ tuples of distinct codewords $\vc_1,\ldots,\vc_L$ in $C$ such that for all $\bu\in [q]^L$ 
$$
|\type_\vu(\bc_1,\ldots,\vc_L)-q^{-L}|\leq \epsilon
$$
and so we have 
$$
|\ol{\rad}_{\omega,\ell}(\phi_\ell(\vc_1), \cdots, \phi_\ell(\vc_L))-f_\ell(U_q,\omega)|\leq q^L\epsilon \ .
$$
\end{enumerate}  
\end{theorem}
\begin{proof}
Set $h=\binom{q}{\ell}$. Let $\epsilon$ satisfy $\abs{ \paren{ \frac{1}{q}-(q-\ell)\delta_0 }^L-q^{-L} }\leq \epsilon$ and $\mu=\left(\frac{1}{q}-(q-\ell)\delta_0\right)^L$, $\delta=\mu\delta_0$. Set $M_0(s)=q^Lhs$.
We assume that the first statement does not hold and our goal is to show that the second statement must hold. Since the first statement does not hold, for any $\vY\in \cX^n$, there exists a codeword $\vc\in \cC$ with $\distlr(\vc,\vY)>n\paren{ 1-\frac{\ell}{q}-\delta }$. For each $\vc\in \cC$, let $T_\vc\in \cX$ be the collection of the $\ell$ most frequent symbols. By the pigeonhole principle, we can find a subcode $\cC'\subseteq \cC$ of size at least $\frac{M}{h}$ such that $T_\vc$ for $\vc\in \cC'$ are the same $T$. Let $\vT=(T,T,\ldots,T)\in \cX^n$
It is clear that $\rad_{\ell}(\cC')\leq \frac{1}{n}\max_{\vc\in \cC'} \distlr(\vc, \vT).$ As $M>hs$, this implies $\distlr(\vc_1, \vT)>\paren{ 1-\frac{\ell}{q}-\delta }n$ for some $\vc_1\in \cC'$. (In fact, there exist at least $M-hs$ such $\vc_1$ as we can remove $\vc_1$ from $\cC'$ and obtain the same conclusion.)
Note that necessarily $\distlr(\vc_1, \vT)\leq \left(1-\frac{\ell}{q}\right)n$ (otherwise $T$ would not be the $\ell$-element subset agreeing the most with $\vc_1$). Let $A_x=\{i\in [n]: \vc_1(i)=x\}$ for $x\in [q]$. This implies 
$$\frac{|A_x|}{n}\in \brack{ \frac{1}{q}-(q-\ell)\delta, \frac{1}{q}+\delta }\subseteq \brack{ \frac{1}{q}-(q-\ell)\delta_0, \frac{1}{q}+\delta_0 },$$
as $\max_{x\in [q]} \frac{|A_x|}{n}\in \left[\frac{1}{q},\frac{1}{q}+\delta\right]$ and $\min_{x\in [q]} \frac{|A_x|}{n}\in \left[\frac{1}{q}-(q-\ell)\delta,\frac{1}{q}\right]$.  
 
Now we fix $\vc_1$ and its index set $A_1,\ldots,A_q$ and let $\cC'=\cC\setminus\{\vc_1\}$. We consider the punctured code $\pi_{A_1}(\cC)$. According to \Cref{lm:lrkey}, there exists a subcode $\cC''\subseteq \cC$ of size at most $hs-1$ with $\rad_{\ell}(\pi_{A_1}(\cC''))\leq 1-\frac{\ell}{q}-\delta_0$. Therefore, the same argument as above shows that there exists at least $M-2hs$ codewords $\vc_2\in \cC$ such that the symbol distribution of $\pi_{A_1}(\vc_2)$ is close to uniform, i.e., $|\{i\in A_1: \vc_2(i)=x\}|/|A_1|\in \brack{ \frac{1}{q}-(q-\ell)\delta_0, \frac{1}{q}+\delta_0 }$ for each $x \in [q]$. 
Then, we apply this argument with sets $A_2,\ldots,A_q$ sequentially and conclude that there exists at least $M-2qhs-1$ (excluding $\vc_1$) codewords $\vc_2\in \cC$ such that the symbol distribution of each $\pi_{A_x}(\vc_2)$ is close to uniform. 

We next partition $[n]$ into $q^2$ sets $A_{xy}=\{i\in [n]: \vc_1(i)=x,\vc_2(j)=y\}$ for $x,y \in [q]$ according to the value of $\vc_1$ and $\vc_2$. This gives $\frac{|A_{xy}|}{n}\in \brack{ \paren{ \frac{1}{q}-(q-\ell)\delta_0 }^2, \paren{ \frac{1}{q}+\delta_0 }^2 }$.
One can continue this process and construct $L$-tuples $\vc_1,\ldots,\vc_L$ for which necessarily
\begin{align} \label{eq:random-like-lr}
    \forall \vu \in [q]^L,~~~|\type_\vu(\bc_1,\ldots,\vc_L)-q^{-L}|\in \brack{ \left(\frac{1}{q}-(q-\ell)\delta_0\right)^L, \paren{ \frac{1}{q}+\delta_0 }^L }
\end{align}
In general, there are more than 
$$
N_1=\prod_{i=0}^{L-1}(M-j-2q^ihs)
$$
$L$-tuples $(\bc_1,\ldots,\vc_L)$ satisfying \Cref{eq:random-like-lr}. This implies $N_1\geq M^L-cM^{L-1}$ where $c$ only depends on $q$, $\ell$ and $L$. The proof is completed.
\end{proof}

%% file: list-rec-end.tex
The proof is quite similar to the list-decoding case. We provide the proof for completeness.  Define $\rho_{\ell}(\cC)=\min \rad_{\ell}(\phi_{\ell}(\vc_1),\ldots,\phi_{\ell}(\vc_L))$ with minimum taken over all $L$-tuples $(\vc_1,\ldots,\vc_L)\in C^L$ with distinct codewords. 

\begin{theorem} \label{thm:list-rec-main}
Let $L\geq 2$ and $q,L>\ell\geq 2$. If $\cC\subseteq [q]^n$ is $(p_*(q,\ell,L)+\epsilon,\ell, L)$-list recoverable, then $|\cC|=  O_{q,\ell,L}(\frac{1}{\epsilon})$. 
\end{theorem}
\begin{proof}
To simplify our notation, let $\tau_{\ell,L}=p_*(q,\ell,L)$ and $\tau_{x,\ell,L}=f_{\ell}(P_{q,\ell,x},U_L)$ with $P_{q,\ell,x}$ defined in \Cref{prop:lrschur_convex}.
Similar to the list-decoding case, by defining $\cC_1$ to be the set of all $\vc \in \cC$ whose first $rL$
coordinates are some given string in $[q]^{rL}$ for $r = \lfloor 1/(\tau_L+\eps)\rfloor$, we can obtain a code $\cC_2\subseteq [q]^{n-rL}$ of size at least $q^{-rL}|\cC|$ whose $\ell$-radius is at least $\tau_{\ell,L}+\epsilon+\frac{L}{n}$. 
Applying \Cref{lm:lrconvert}, we thus have a subcode $\cC_2\subseteq \cC$ with $\rho_{\ell}(\cC_2)\geq \tau_{\ell,L}+\epsilon$. 
We divide our discussion into two cases. 
\begin{itemize}
\item \label{ramsey-listrec} $\rad_{\ell}(\cC_2)\leq 1-\frac{\ell}{q}-\delta$ for some constant $\delta>0$. Let $p=\frac{\ell}{q}+\delta$. 
For every $(\vc_1,\ldots,\vc_L)\in \cC_1^{L}$, there exists a weight $\omega\in \Omega_{\ell,L}$ that solves 
$$\rad(\phi_{\ell}(\vc_1),\ldots,\phi_{\ell}(\vc_L))=\max_{\omega\in \Omega_{\ell,L}}\ol{\rad}_{\omega,\ell}(\phi_{\ell}(\vc_1), \cdots, \phi_{\ell}(\vc_L)).$$ Each solution $\omega$ gives a coloring of $L$-element subsets of $\cC_2$. By \Cref{lm:lrfinite}, there are a finite number of $\omega$ in $\Omega_{\ell,L}$. The hypergraph version of Ramesy's theorem~\cite[Theorem~2]{graham1991ramsey} implies that there exists a monochromatic subset $\cC_3\subseteq \cC_2$ exceeding $\frac{L^2}{\tau_{\ell,L}-\tau_{p,\ell,L}}$. 

On the other hand, let $\cT$ be the set of all $L$-tuples with distinct codewords in $\cC_3$.  Let $(\vc_1,\ldots,\vc_L)$ be an $L$-tuple selected uniformly at random in $\cC_3^{L}$. Then 
$$\Pr[(\vc_1,\ldots,\vc_L)\notin \cT]\leq \frac{\binom{L}{2}}{|C_3|}<\tau_{\ell,L}-\tau_{p,\ell,L}.$$ 
Since $\rad_\ell(\cC)\leq 1-\frac{\ell}{q}-\delta$, by \Cref{lm:lrconvexity} and \Cref{prop:lrschur_convex}, we have
\begin{eqnarray*}
 \tau_{p,\ell,L}&\geq& \exptover{(\vc_1,\ldots,\vc_L)\in \cC_3^L}{\ol{\rad}_{\omega,\ell}(\phi_{\ell}(\vc_1), \cdots, \phi_{\ell}(\vc_L))}\\
 &\geq& \Pr[(\vc_1,\ldots,\vc_L)\in \cT]\exptover{(\vc_1,\ldots,\vc_L)\in \cT}{\ol{\rad}_{\omega,\ell}(\phi_{\ell}(\vc_1), \cdots, \phi_{\ell}(\vc_L))}.
 \end{eqnarray*}
This implies that there exists an $L$-tuple of distinct codewords $\vc_1,\ldots,\vc_L$ in $\cC_3$ such that $$(1-\tau_{\ell,L}+\tau_{p,\ell,L})\ol{\rad}_{\omega,\ell}(\phi_{\ell}(\vc_1), \cdots, \phi_{\ell}(\vc_L))< \tau_{p,\ell,L}.$$ 
It follows that 
$$
\ol{\rad}_{\omega,\ell}(\phi_{\ell}(\vc_1), \cdots, \phi_{\ell}(\vc_L))<\tau_{p,\ell,L}+\tau_{\ell,L}-\tau_{p,\ell,L}=\tau_{\ell,L}
$$
contradicting the assumption that $\rho_L(\cC_2) \geq \tau_{\ell,L}+\eps$. 
\item Otherwise, let $\cH$ be the collection of all $L$-tuples $(\vc_1,\ldots,\vc_L)$ in $\cC_2^L$ such that $\ol{\rad}_{\omega,\ell}(\phi_{\ell}(\vc_1), \cdots, \phi_{\ell}(\vc_L))>\tau_{\ell,L}$ for some $\omega\neq U_L$. Let $\epsilon_0=q^{-L}\min\{\tau_{\ell,L}-f_\ell(U_q,\omega): \omega\in \Omega_{\ell,L}\}$; \Cref{thm:lrmaxq} and \Cref{lm:lrfinite} guarantee $\eps_0>0$. 
By \Cref{thm:lrradius}, there exists at least $|\cC_2|^L-c|\cC_2|^{L-1}$ many $L$-tuples of distinct codewords $\vc_1,\ldots,\vc_L$ in $\cC_2$ such that 
$$\ol{\rad}_{\omega,\ell}(\phi_\ell(\vc_1), \cdots, \phi(\vc_L))\leq f_\ell(U_q,\omega)+q^L\epsilon_0\leq \tau_{\ell,L}.$$
Thus, $|\cH|\leq c |\cC_2|^{L-1}$ where $c$ only depends on $q,\ell,L$. Let $(\bc_1,\ldots,\bc_L)$ be a random $L$-tuple in $\cC_2^L$. Similarly, we can show
$\Pr[(\vc_1,\ldots,\vc_L)\in \cT\setminus \cH]\geq 1-O(\frac{1}{|\cC_2|})$ where the constant in $O$ only depends on $q,\ell,L$ and 
\begin{eqnarray*}
\tau_{\ell,L}\geq\paren{ 1-O\paren{\frac{1}{|\cC_2|}} }\exptover{(\vc_1,\ldots,\vc_L)\in \cT\setminus \cH}{\ol{\rad}_{U_L,\ell}(\phi_\ell(\vc_1), \cdots, \phi_\ell(\vc_L))}
 \end{eqnarray*}
On the other hand, 
$$\ol{\rad}_{U_L,\ell}(\phi_\ell(\vc_1), \cdots, \phi_\ell(\vc_L))\geq \rho_\ell(\cC_1)\geq \tau_{\ell,L}+\epsilon.$$
for $(\vc_1,\ldots,\vc_L)\in \cT\setminus\cH$.
This implies that $|\cC_2|\leq O_{q,\ell,L}(\frac{1}{\epsilon})$ and thus $|\cC|\leq  O_{q,\ell,L}(\frac{1}{\epsilon})$.

\end{itemize}
\end{proof}

%% file: list-dec-constr.tex

In this section, we present a simple simplex-like code construction and show that it attains the optimal size-radius trade-off by analyzing its list-decoding and -recovery radius. 

Our construction will be identical for list-decoding and -recovery and therefore we will directly analyze its list-recovery radius. 
Before presenting the construction and its analysis, let us define the average radius $ \ol{\rad}_\ell $. 
This is a standard notion that ``linearizes'' the Chebyshev radius $ \rad $ and often finds its use in the analysis of list-recoverable codes in the literature. 
The definition reads as follows: for any $ \vc_1, \cdots, \vc_L\in[q]^n $, 
\begin{align}
    \ol{\rad}_\ell(\vc_1, \cdots, \vc_L) &\coloneqq \frac{1}{L}\min_{\vY\in \cX^n} \sum_{i = 1}^L \distlr(\vc_i, \vY) . \notag 
\end{align}
It is well-known and easy to verify (by, e.g., following the derivations leading to \Cref{eqn:wt_avg_rad_2}) that the above minimization admits the following explicit solution: 
\begin{align}
    \ol{\rad}_\ell(\vc_1, \cdots, \vc_L) &\coloneqq \sum_{j = 1}^n \paren{ 1 - \frac{1}{L} \plur_\ell(\vc_1(j), \cdots, \vc_L(j)) } , \label{eqn:def-avgrad-plur} 
\end{align}

\Cref{eqn:def-avgrad-plur} should be interpreted as the average distance from each $\vc_i$ to the ``centroid'' $ \vY^*\in \cX^n $ of the list defined as\footnote{If there are multiple maximizers, take an arbitrary one and the value of $ \ol{\rad}_\ell $ remains the same. }
\begin{align}
    \vY^*(j) &\coloneqq \argmax_{A\in \cX} \sum_{i = 1}^L \indicator{ \vc_i(j)\in A } . \notag 
\end{align}
for each $j\in[n]$. 

Finally, for integers $q\ge1$ and $L\ge0$, denote by
\begin{align}
    \cA_{q, L} &= \brace{ (a_1, \cdots, a_q) \in \bbZ_{\ge0}^q : \sum_{i = 1}^q a_i = L } \notag 
\end{align}
the set of $q$-partitions of $L$, i.e., $a_i$ is the number of indices taking value $i$. For $\va\in \cA_{a,L}$, we shorthand $\binom{L}{\va}=\binom{L}{a_1,\ldots,a_q}$ where $\va=(a_1,\ldots,a_q)$. Define $\maxl\{\va\}=\max_{A\in \cX}\sum_{i\in A}a_i$, i.e., the sum of $\ell$ largest components in $\va$. 
\begin{theorem}[Construction of zero-rate list-recoverable codes]
\label{thm:list-dec-qary-postplotkin-constr}
Fix any integers $ q\ge3 $, $\ell\geq 1$ and  $L\ge2 $. 
For any sufficiently large $m$, there exists a $(p,\ell,L)$-list-recoverable code $\cC$ with blocklength 
\begin{align}
    n = \binom{qm}{\LaTeXunderbrace{m,\cdots,m}_{q}} , \label{eqn:blocklength}
\end{align}
and the trade-off between code size $M$ and (relative) radius $p$ given by:
\begin{align}
    M = qm , \quad 
    p = p_*(q, \ell, L) + c_{q, \ell, L} m^{-1} + O(m^{-2}) , \notag 
\end{align}
where 
\begin{align}
c_{q,\ell,L} &\coloneqq q^{-L} \sum_{\va\in\cA_{q,L}} \frac{\maxl\brace{\va}}{L} \binom{L}{\va} \brack{\sum_{i = 1}^q\binom{a_i}{2} - \frac{1}{q}\binom{L}{2}} > 0 . \label{eqn:const_constr} 
\end{align}
\end{theorem}

\begin{proof}
Let $ m\in\bbZ_{\ge1} $ be sufficiently large. 
Consdier the following codebook $\cC$ of size $ M\times n $ where $ M = qm $ and $n$ given in \Cref{eqn:blocklength}. 
This codebook $\cC$ as an $M$-by-$n$ matrix consists of all possible length-$qm$ vectors with $ m $ ones, $m$ twos, ..., and $m$ $q$'s as its columns. 
Each row forms a codeword. 
Recall that $ \rad_\ell \ge \ol{\rad}_\ell $. 
Therefore, to show list-decodability, it suffices to lower bound $ \ol{\rad}_\ell $. 
By symmetry, 
$ \ol{\rad}_\ell(\cL) $ is independent of the choice of $ \cL\in\binom{\cC}{L} $, so it is equivalent to compute $ \ol{\rad}_\ell(\cL) $ averaged over $ \cL\in\binom{\cC}{L} $. 
Recall from \Cref{eqn:def-avgrad-plur} that $ \ol{\rad}_\ell(\cL) $ can be decomposed as the sum of average radii of each column of $ \cL $ (viewed as an $L$-by-$n$ matrix). 
By symmetry, averaged over $\cL$, the average radius of each column is the same which is equal to 
\begin{align}
1 - \frac{1}{L} \expt{\plur_\ell(X_1,\cdots,X_L)} , \notag 
\end{align}
where $ (X_1,\cdots,X_L) $ is a uniformly random $L$-sub(multi)set of 
\begin{align}
(\LaTeXunderbrace{1,\cdots,1}_{m},\LaTeXunderbrace{2,\cdots,2}_{m},\cdots,\LaTeXunderbrace{q,\cdots,q}_{m}). \label{eqn:constr_randomness} 
\end{align}

For $ a_1, \cdots, a_{q-1} \in\bbZ_{\ge0}^{q-1} $ such that $ a_1 + \cdots + a_{q-1}\le L $, 
let
\begin{align}
    \binom{L}{a_1, \cdots, a_{q-1}, \star} &\coloneqq \binom{L}{a_1, \cdots, a_{q-1}, L-\sum_{i = 1}^{q-1} a_i} . \notag 
\end{align}
Now let us compute
\begin{align}
    &\phantom{=}~ \frac{1}{L} \expt{\pl_\ell(X_1, \cdots, X_L)} \notag \\
    &= \frac{1}{L} \sum_{\substack{(a_1, \cdots, a_\ell)\in\bbZ^\ell \\ \forall i\in[\ell], \ceil{\ell L/q}\le a_i\le L}} \paren{\sum_{i = 1}^\ell a_i} \cdot \frac{\binom{q}{\ell}}{\binom{qm}{m, \cdots, m}} \cdot \binom{L}{a_1, \cdots, a_\ell, \star} \binom{qm - L}{m - a_1, \cdots, m-a_\ell, \star} \notag \\
    &\phantom{=}~ \times \sum_{\substack{(a_{\ell+1}, \cdots, a_q)\in\bbZ^{q - \ell} \\ \forall\ell+1\le i\le q, 0\le a_i\le \min\{a_1, \cdots, a_\ell\} \\ a_{\ell+1} + \cdots + a_q = L - (a_1 + \cdots + a_\ell)}} \binom{L - (a_1 + \cdots + a_\ell)}{a_{\ell+1}, \cdots, a_q} \binom{am - L - (m-a_1) - \cdots - (m-a_\ell)}{m-a_{\ell+1}, \cdots, m-a_q} \notag \\
    &= \sum_{\substack{(a_1, \cdots, a_\ell)\in\bbZ^\ell \\ \forall i\in[\ell], \ceil{\ell L/q}\le a_i\le L}} 
    \sum_{\substack{(a_{\ell+1}, \cdots, a_q)\in\bbZ^{q - \ell} \\ \forall\ell+1\le i\le q, 0\le a_i\le \min\{a_1, \cdots, a_\ell\} \\ a_{\ell+1} + \cdots + a_q = L - (a_1 + \cdots + a_\ell)}}
    \binom{q}{\ell} \frac{\sum_{i = 1}^\ell a_i}{L} 
    \cdot \binom{L}{a_1, \cdots, a_\ell, \star} 
    \binom{L - (a_1 + \cdots + a_\ell)}{a_{\ell+1}, \cdots, a_q} 
    \notag \\
    &\phantom{=}~ \times \binom{qm - L}{m - a_1, \cdots, m-a_\ell, \star} 
    \binom{am - L - (m-a_1) - \cdots - (m-a_\ell)}{m-a_{\ell+1}, \cdots, m-a_q} 
    \cdot \binom{qm}{m, \cdots, m}^{-1} \notag \\ 
    &= \sum_{\substack{(a_1, \cdots, a_q)\in\bbZ_{\ge0}^q \\ a_1 + \cdots + a_q = L}} \frac{\maxl\brace{a_1, \cdots, a_q}}{L} \binom{L}{a_1, \cdots, a_q} \binom{qm - L}{m - a_1, \cdots, m - a_q} \binom{qm}{m, \cdots, m}^{-1} \notag 
\end{align}
Taking the Taylor expansion at $ m\to\infty $, it can be computed that
\begin{align}
\binom{qm - L}{m -  a_1,\cdots,m -  a_{q}} \binom{qm}{m,\cdots,m}^{-1} 
&= \frac{(qm - L)!}{(m- a_1)!\cdots(m -  a_q)!}\frac{m!\cdots m!}{(qm)!} \notag \\
&= \frac{\prod_{i_1 = 0}^{ a_1 - 1}(m - i_1)\prod_{i_2 = 0}^{ a_2 - 1}(m - i_2)\cdots\prod_{i_q = 0}^{ a_q - 1}(m - i_q)}{(qm)(qm - 1)\cdots(qm - L+1)} \notag \\
&= \frac{\prod_{i_1 = 1}^{ a_1 - 1}(1 - i_1m^{-1})\prod_{i_2 = 1}^{ a_2 - 1}(1 - i_2m^{-1})\cdots\prod_{i_q = 1}^{ a_q - 1}(1 - i_qm^{-1})}{q(q - m^{-1})\cdots(q - (L-1)m^{-1})} \notag \\
&= q^{-L}\brack{1 + \frac{1}{m}\paren{\frac{1}{q}\binom{L}{2} - \sum_{i = 1}^q\binom{ a_i}{2}} + O\paren{\frac{1}{m^2}}}. \notag 
\end{align}
Recall 
\begin{align*}
    f_{q,L,\ell}(P) &\coloneqq \exptover{(X_1,\cdots,X_L)\sim P^{\ot L}}{\plur_\ell(X_1,\cdots,X_L)} .
\end{align*}
and $p_*(q,\ell,L)=1-f_{q,L,\ell}(U_q)$.
Therefore, 
\begin{align}
    &\phantom{=}~ \frac{1}{L} \expt{\pl_\ell(X_1, \cdots, X_L)} \notag \\
    &= q^{-L} \sum_{\substack{(a_1, \cdots, a_q)\in\bbZ_{\ge0}^q \\ a_1 + \cdots + a_q = L}} \frac{\maxl\brace{a_1, \cdots, a_q}}{L} \binom{L}{a_1, \cdots, a_q} \brack{1 + \frac{1}{m}\paren{\frac{1}{q}\binom{L}{2} - \sum_{i = 1}^q\binom{ a_i}{2}} + O\paren{\frac{1}{m^2}}} \notag \\
    &= \frac{1}{L} f_{q,L,\ell}(U_q) - \frac{1}{m} q^{-L} \sum_{\va\in\cA_{q,L}} \frac{\maxl\brace{\va}}{L} \binom{L}{\va} \brack{\sum_{i = 1}^q\binom{a_i}{2} - \frac{1}{q}\binom{L}{2}} + O(m^{-2}) \notag \\
    &= L^{-1} f_{q,L,\ell}(U_q) - c_{q,\ell,L} m^{-1} + O(m^{-2}) . \notag 
\end{align}
Then we have 
\begin{align}
1 - L^{-1} f_{q,L,\ell}(U_q) + c_{q,\ell,L} m^{-1}  + O(m^{-2}) 
= p_*(q,\ell,L) + c_{q,\ell,L} m^{-1}  + O(m^{-2}) , \notag 
\end{align}
To complete the proof, it remains to verify that $c_{q,\ell,L}$ is always positive.
This is equivalent to showing
$$
\frac{L c_{q,L}}{q^L\binom{L}{2}}=\sum_{\va\in\cA_{q,L}} \maxl\brace{\va} \binom{L}{\va} \left(\sum_{i = 1}^q \frac{\binom{a_i}{2}}{\binom{L}{2}} - \frac{1}{q}\right)> 0 
$$
If $L=2$, we have
$$
\sum_{\va\in\cA_{q,L}} \maxl\brace{\va} \binom{2}{\va} \left(\sum_{i = 1}^q \binom{a_i}{2} - \frac{1}{q}\right)=q(q-1)\times\paren{-\frac{1}{q}}+2q\times\paren{1-\frac{1}{q}}>0
$$
In what follows, we assume $L>2$.
We note that  $\binom{L}{\va}\binom{a_i}{2}/\binom{L}{2}=\binom{L-2}{\va-2\ve_i}$ where 
$$\ve_i=(\LaTeXunderbrace{0,\cdots,0}_{i-1},1,\LaTeXunderbrace{0,\cdots,0}_{q-i-1}).$$ 
We abuse the notation by letting $\binom{L-2}{\va-2\ve_i}=0$ if $a_i=0,1$. 

\begin{align}
\sum_{\va\in\cA_{q,L}} \maxl\brace{\va} \binom{L}{\va} \sum_{i = 1}^q \frac{\binom{a_i}{2}}{\binom{L}{2}} &= \sum_{\va\in\cA_{q,L}} \maxl\brace{\va} \sum_{i = 1}^q \binom{L-2}{\va-2\ve_i} \notag \\
&= \sum_{\va\in\cA_{q,L-2}}  \sum_{i = 1}^q \maxl\brace{\va+2\ve_i}\binom{L-2}{\va} . \notag
\end{align}
On the other hand, we have 
$$
\frac{1}{q}\sum_{\va\in\cA_{q,L}} \maxl\brace{\va} \binom{L}{\va}=\frac{1}{q}\sum_{\va\in\cA_{q,L-2}}\sum_{i,j=1}^q \maxl\brace{\va+\ve_i+\ve_j}\binom{L-2}{\va}.
$$
This is because we can separate the $L$ symbols into two sets the first set containing $L-2$ symbols and the second one containing $2$ symbols. $\binom{L-2}{\va}$ represents the number of ways to select $L-2$ symbols from $[q]$ so that it produces $\va$. Then, we can pick the last two symbols from $[q]$ in an arbitrary manner which results in $\va+\be_i+\be_j$. 
Note that $\maxl\brace{\va+2\ve_i}+\maxl\brace{\va+2\ve_j}\geq 2\maxl\brace{\va+\ve_i+\ve_j}$.
Since the equality does not hold for every $\va\in \cA_{q,L-2}$ if $L>\ell$, we conclude $c_{q,\ell,L}>0$. 
\end{proof}